\documentclass[a4paper,UKenglish,cleveref, autoref, thm-restate,authorcolumns]{lipics-v2019}

\title{Decomposition Polygons into Fat Components} 

\author{Maike Buchin}{Ruhr University Bochum, Germany}{maike.buchin@rub.de}{}{}

\author{Leonie Selbach}{Ruhr University Bochum, Germany}{leonie.selbach@rub.de}{}{}

\authorrunning{M. Buchin and L. Selbach} 

\Copyright{Maike Buchin and Leonie Selbach} 

\ccsdesc[100]{Theory of computation~Computational geometry} 

\keywords{Polygon decomposition, fatness, aspect ratio} 

\nolinenumbers 

\hideLIPIcs  

\EventEditors{John Q. Open and Joan R. Access}
\EventNoEds{2}
\EventLongTitle{42nd Conference on Very Important Topics (CVIT 2016)}
\EventShortTitle{CVIT 2016}
\EventAcronym{CVIT}
\EventYear{2016}
\EventDate{December 24--27, 2016}
\EventLocation{Little Whinging, United Kingdom}
\EventLogo{}
\SeriesVolume{42}
\ArticleNo{23}

\begin{document}

\maketitle

\begin{abstract}
We study the problem of decomposing (i.e. partitioning or covering) polygons into components that are $\alpha$-fat, which means that the aspect ratio of each subpolygon is at most $\alpha$. We consider decompositions without Steiner points. We present a polynomial-time algorithm for simple polygons that finds the minimum $\alpha$ such that an $\alpha$-fat partition exists. Furthermore, we show that finding an $\alpha$-fat partition or covering with minimum cardinality is NP-hard for polygons with holes. 
\end{abstract}

\section{Introduction}
A decomposition of a polygon $P$ is a set of subpolygons whose union is exactly $P$. If the subpolygons are not allowed to overlap, the set is a \emph{partition} of $P$. Otherwise, we call the set a \emph{covering}. Here, we consider decompositions without Steiner points. Thus, a polygon is decomposed by adding diagonals between its vertices. Polygon decomposition problems arise in many theoretical and practical applications and can be categorized with regard to the type of subpolygon that is used~\cite{keil2000polygon}. 
Our research is motivated by a bioinformatical application, namely the processing of tissue samples for molecular analysis using laser capture microdissection (LCM). The size and shape of the tissue fragments are two of the main factors that affect the success of the dissection. The fragmentation of tissue samples into feasible regions can be modeled as a polygon decomposition problem~\cite{selbach_et_al:LIPIcs:2020:12802}. Preliminary results show that
fatness is a suitable shape criterion for this practical application, as specifically round shapes are desirable (see Fig.~\ref{fig:example}). Thus, we study decompositions into fat components.\par
\begin{figure}[b]
    \centering
    \includegraphics[width=11cm]{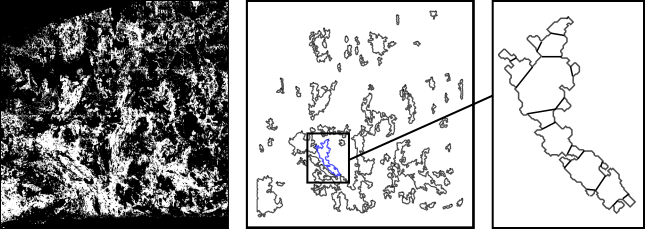}
    \caption{Fat decomposition of one connected component in a preprocessed tissue sample.}
    \label{fig:example}
\end{figure}
A polygon $P$ is called \emph{$\alpha$-fat} if its aspect ratio (AR) is at most $\alpha$. There are different definitions for the aspect ratio and in this paper we consider the following two (see Fig.~\ref{fig:partition_example}):
\begin{description}
    \item[square-fatness:] $AR_{square}$ = ratio between the side length of the smallest axis-parallel square containing $P$ and side length of the largest axis-parallel square contained in $P$~\cite{katz19973}.
    \item[disk-fatness:] $AR_{disk}$ = ratio between the diameter of the smallest circle enclosing $P$ (minimum circumscribed circle or MCC) and the diameter of the largest circle enclosed in $P$ (maximum inscribed circle or MIC)~\cite{damian2004computing}.
\end{description}
A polygon $P$ is called \emph{$\alpha$-small} if the side length of the enclosing square or respectively the diameter of the enclosing circle is at most $\alpha$.
The \emph{minimum $\alpha$-fat partition} (or \emph{covering}) \emph{problem} is finding a partition (or covering) with minimum cardinality such that every subpolygon is $\alpha$-fat. We have analog problems with $\alpha$-smallness instead of $\alpha$-fatness.\par 
\begin{figure}[t]
	
	\centering
	\begin{subfigure}[t]{0.45\textwidth}
	    \centering
		\includegraphics[width=0.8\textwidth]{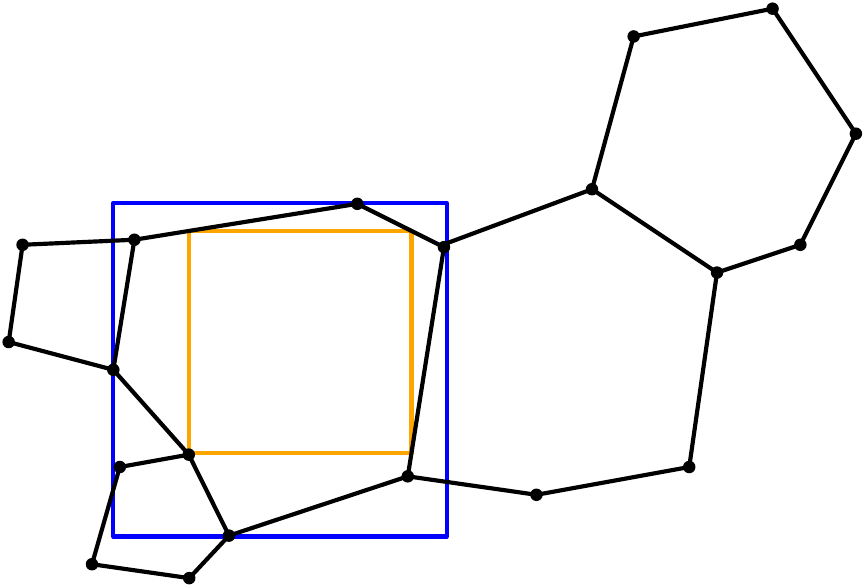}
		\caption{square-fatness}
		\label{fig:square_fat}
	\end{subfigure}\qquad
	\begin{subfigure}[t]{0.45\textwidth}
	    \centering
		\includegraphics[width=0.8\textwidth]{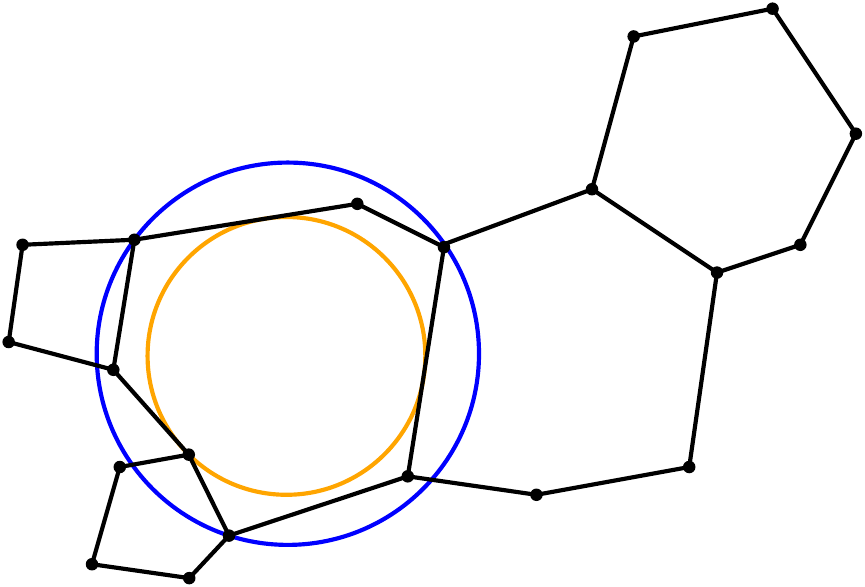}
		\caption{disk-fatness}
		\label{fig:disk_fat}
	\end{subfigure}
	\caption{Comparison of fatness definitions. The partition is $1.5$-fat in (a) and $1.4$-fat in (b).}
	\label{fig:partition_example}
\end{figure}

An overview of related results is presented in Table~\ref{tab:related}. Worman showed that the minimum $\alpha$-small partition problem as well as the covering problem are NP-hard for polygons with holes~\cite{worman2003decomposing}. Square-smallness was used for the reduction but with some adjustments the construction holds for disk-smallness as well (see appendix). For all following results disk-smallness and -fatness was used.
Damian and Pemmaraju showed that the minimum $\alpha$-small partition problem is polynomial-time solvable for simple polygons and that a faster 2-approximation algorithm exists~\cite{damian2004computing}. Additionally, the authors presented an approximation algorithm for convex polygons. Damian proved that the minimum $\alpha$-fat partition problem can be solved in polynomial time for simple polygons and conjectured that this problem is NP-hard for polygons with holes~\cite{damian2004exact}. The \emph{min-fat partition problem} is finding the smallest $\alpha$ for which an $\alpha$-fat partition exists. Solving the min-fat partition problem was stated as an open problem by Damian~\cite{CCCG2002Open}. This problem is of particular interest for us, as the corresponding partition is the most suitable for our practical application.\par 
\begin{table}[]
\caption{Overview of results for fat decomposition problems using different smallness and fatness definitions. Results marked with * are presented in this paper.}
\label{tab:related}
\begin{tabular}{l|ll}
Problems                                         & Simple polygons             & Polygons with holes \\ \hline
Minimum $\alpha$-small partition &   $\mathcal{O}(n^3m)$ (disk)~\cite{damian2004computing}                         & NP-hard (square~\cite{worman2003decomposing} and disk*)           \\
Minimum $\alpha$-small covering      & ?                           & NP-hard   (square~\cite{worman2003decomposing} and disk*)           \\
Minimum $\alpha$-fat partition   & $\mathcal{O}(n^4m^3)$ (disk)~\cite{damian2004exact}                            & NP-hard  (square and disk)*         \\
Minimum $\alpha$-fat covering       & ?                           & NP-hard (square and disk)*         \\
Min-fat partition                & $\mathcal{O}(n^3m^5\log n)$ (disk)* &     ?               
\end{tabular}
\end{table}
This paper includes two main results. In Section~\ref{sec:MinFatPartition} consider the min-fat partition problem using disk-fatness and present a polynomial-time algorithms for simple polygons. In Section~\ref{sec:minimum alpha fat decomposition} we consider the minimum $\alpha$-fat partition and covering problem. We confirm the conjecture that these problems are NP-hard for polygons with holes and present the two reductions for square- as well as disk-fatness. Furthermore, we present the NP-hardness reduction of the minimum $\alpha$-small partition problem for disk-smallness in Appendix~\ref{app:disk reduction}.\par  

\section{Min-fat partition problem for simple polygons}
\label{sec:MinFatPartition}
For our application, the main goal is to minimize tissue loss due to unsuccessful dissections and the shape of the fragments is one of the key factors. Our goal is to find an optimal partition of a given polygon such that the largest aspect ratio $AR$ (regarding disk-fatness) of any subpolygon is minimized. The value of the largest $AR$ in an optimal partition equals the desired $\alpha$ in the min-fat partition problem. With our algorithm, we extend the method of Damian~\cite{damian2004exact} for this optimization problem. \par
Let $P$ be a simple polygon with vertices labeled from $1$ to $n$ counterclockwise. A diagonal $(i,j)$ is a line segment that connects two vertices $i$ and $j$ and does not intersect the outside of $P$. Let $G(P)$ be the visibility graph of $P$ consisting of the $n$ vertices of $P$ and $m$ diagonals. We define $S$ as the set consisting of all vertices and edges of $G(P)$. For each diagonal $(i,j)$ with $i<j$, let $P_{i,j}$ be the subpolygon with vertices $\{i,i+1,\ldots,j\}$ (see Fig.~\ref{fig:polygon}). To solve the min-fat partition problem, we compute an optimal partition $Z_{i,j}$ for each $P_{i,j}$. This can be done iteratively. Let $Q$ be the polygon in $Z_{i,j}$ adjacent to $(i,j)$. Note that the vertices of $Q$ induce a path $q$ from $i$ to $j$ in the visibility graph and we have $Z_{i,j} = \bigcup_{(k,l)\in q}Z_{k,l} \cup Q$ (see Fig.~\ref{fig:path}). The idea is to find an optimal partition of $P_{i,j}$ by computing the optimal path $q$. We define edge weights $w(i,j)$ as the value of an optimal partition $Z_{i,j}$ of $P_{i,j}$:
\begin{align*}
    w(i,j)=\max_{P'\in Z_{i,j}}AR(P') = \max\{ \max_{(k,l)\in q}w(k,l),AR(Q)\}
\end{align*}
for $AR(Q)=d(\text{MCC})/d(\text{MIC})$ being the ratio between the diameters $d(\cdot)$ of the minimum circumscribed circle MCC and maximum inscribed circle MIC of $Q$. If $j$ is equal to $i+1$, the partition $Z_{i,i+1}$ is empty and we set $w(i,i+1)=0$. Otherwise, $w(i,j)$ equals the value of the largest AR in an min-fat partition of $P_{i,j}$.\par 
\begin{figure}[b]
	
	\centering
	\begin{subfigure}[t]{0.3\textwidth}
	    \centering
		\includegraphics[width=\textwidth]{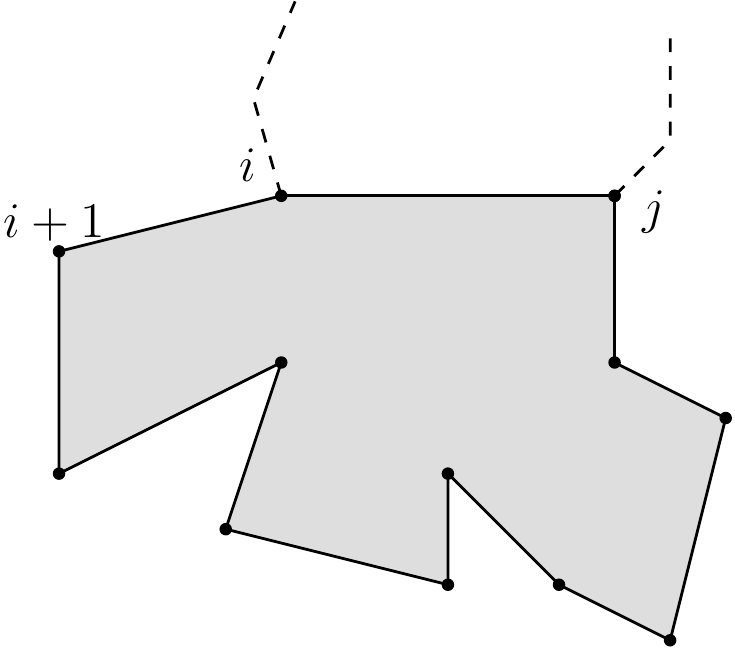}
		\caption{Polygon $P_{i,j}$.}
		\label{fig:polygon}
	\end{subfigure}\hfill
	\begin{subfigure}[t]{0.3\textwidth}
	    \centering
		\includegraphics[width=\textwidth]{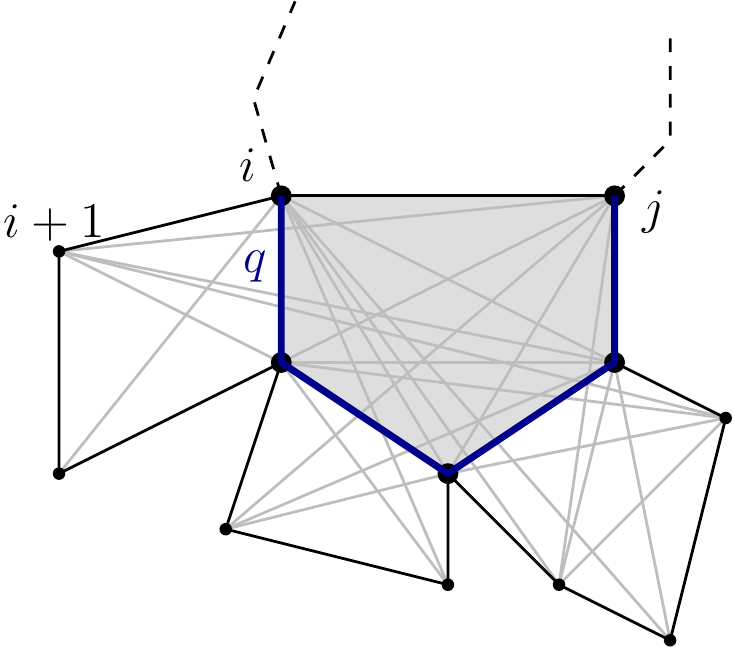}
		\caption{Polygon $Q$.}
		\label{fig:path}
	\end{subfigure}\hfill
	\begin{subfigure}[t]{0.3\textwidth}
	    \centering
		\includegraphics[width=\textwidth]{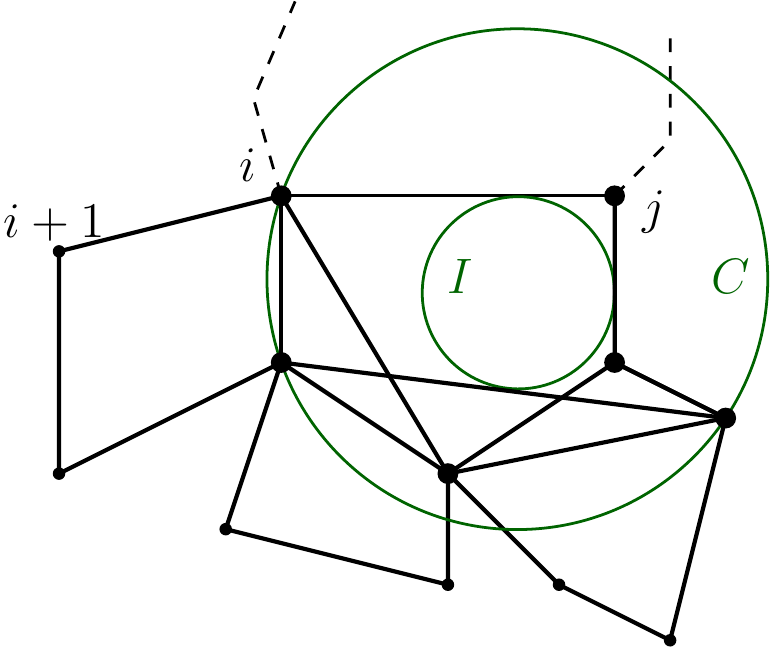}
		\caption{Graph $G_{i,j}^{(C,I)}$.}
		\label{fig:reducedVisGraph}
	\end{subfigure}
	\caption{In (a): Subpolygon $P_{i,j}$. In (b): Polygon $Q$ (gray) in an optimal partition of $P_{i,j}$ induced by a path $q$ (blue edges) in the visibility graph. In (c): Reduced visibility graph $G_{i,j}^{(C,I)}$ (fat edges) for a pair of circles $(C,I)$.}
	\label{fig:Graph}
\end{figure}
However, the computation of $w(i,j)$ presents the following problem: Finding the path $q$ and its correct edges requires knowledge about the resulting polygon $Q$ and its aspect ratio, which is not available beforehand. Therefore, we compute paths on different reduced graphs that ensure that the aspect ratio of each possible polygon is below a certain value and then choose the best one.  
For this, we consider all pairs of circles $(C,I)$ such that the following properties (P) hold: 
\begin{itemize}
    \item $(i,j)$ lies completely inside of $C$ and outside of $I$,
    \item $I$ is tangent to 3 elements in $S$,
    \item and $C$ either touches 3 vertices of $P$ or its diameter connects 2 vertices.
\end{itemize}
 Note that for any subpolygon the pair (MCC,MIC) fulfills these properties. For each $(C,I)$, let $G_{i,j}^{(C,I)}$ be the subgraph of $G(P)$ consisting of edges that lie outside of $I$ and inside of $P_{i,j}$ and $C$ (see Fig.~\ref{fig:reducedVisGraph}).\par 
We can compute the weights $w(i,j)$ by using dynamic programming with increasing values of $j-i$. For each pair of circles $(C,I)$ that fulfill (P), we compute a corresponding weight $W(C,I)$ and use those values to determine an edge weight $w'(i,j)$ as follows:
\begin{align*}
    W(C,I) &= \min_{q \in G_{i,j}^{(C,I)}} \max\{ \max_{(k,l)\in q} w'(i,j), d(C)/d(I)\}\\
    w'(i,j) &= \min_{(C,I)\text{ fulfilling (P)}} W(C,I)
\end{align*}
The weights of all edges except for $(i,j)$ have already been computed. We search in $G_{i,j}^{(C,I)}$ for the path $q$ such that the value $\max\{\max_{(k,l)\in q}w'(k,l), d(C)/d(I)\}$ is minimized. We denote this optimal value as $W(C,I)$. Over all possible combinations of circles, we search for the pair $(C,I)$ with minimum $W(C,I)$ and set $w'(i,j)=W(C,I)$. If no pair of circles exist, we set the weight $w'(i,j)=0$. We can show by induction that $w'(i,j)$ is actually equal to the largest aspect ratio in the corresponding partition and that this partition is indeed optimal and hence $w'(i,j)=w(i,j)$.

\begin{restatable}{lemma}{weight}\label{lem:weight}
For an edge $(i,j)$ in the visibility graph $G(P)$ with $j\neq i+1$, let $w'(i,j)$ be the computed weight and $Z_{i,j}$ the corresponding partition. Then, $w'(i,j)=\max_{P'\in Z_{i,j}}AR(P')$.
\end{restatable}
\begin{proof}
Let $(C',I')$ be the pair for circles, $q'$ the corresponding path in $G_{i,j}^{(C',I')}$ and $(k',l')$ the edge in $q'$ such that
\[
w'(i,j) = W(C',I') = \max\{\max_{(k,l)\in q'}w'(k,l),d(C')/d(I')\} = \max\{w'(k',l'),d(C')/d(I')\}.
\]
The computation induces a partition $Z_{i,j}$ in which we denote the polygon adjacent to $(i,j)$ by $Q'$. By induction over $j-i$, we show that $w'(i,j)=\max_{P'\in Z_{i,j}}AR(P')$.\\
For $j-i=2$, consider an edge $(i,j)=(i,i+2)$ in the visibility graph. There is only one possible pair $(C',I')$, which is equal to (MCC,MIC) of $P_{i,j}$. Thus, there is only one possible path $q'=\{(i,i+1),(i+1,i+2)\}\in G_{i,j}^{(C',I')}$. This path induces the partition $Z_{i,j}=\{P_{i,j}\}$. Therefore, we have $Q'=P_{i,j}$ and $d(C')/d(I') = AR(Q')$. Since $w'(i,i+1)=w'(i+1,i+2)=0$, we have
\[
w'(i,j) = \frac{d(C')}{
d(I')} = AR(Q') = \max_{P'\in Z_{i,j}}AR(P').
\]
Our induction hypothesis is that $w'(i,j)=\max_{P'\in Z_{i,j}}AR(P')$ for all $(i,j)$ such that $j-i\leq s$. Now, we consider an edge $(i,j)$ such that $j-i=s+1$.
Using the induction hypothesis (IH), we have the following:
\begin{align*}
 \max_{P'\in Z_{i,j}}AR(P') &\;\,= \max_{(k,l)\in q}\{\max_{P'\in Z_{k,l}}AR(P'),AR(Q)\}  \underset{(IH)}{=} \max\{w'(k',l'),AR(Q)\}
\end{align*}
Since $d(C')/d(I') \geq AR(Q')$, the inequality $w'(i,j) \geq \max_{P'\in Z_{i,j}}AR(P')$ holds. For the opposite inequality we consider different cases:
\begin{description}
    \item[Case 1:] $w'(k',l') \geq d(C')/d(I')$   ($\geq AR(Q')$).\\
    Obviously, we have $w'(i,j)=w'(k',l')=\max_{P'\in Z_{i,j}}AR(P')$.
    \item[Case 2:] $w'(k',l') < d(C')/d(I')$.\\
    Then, the weight $w'(i,j)$ is equal to $d(C')/d(I')$. 
    \begin{description}
        \item[a)] $d(C')/d(I') = AR(Q')$.\\
        We have $w'(i,j)=AR(Q')=\max_{P'\in Z_{i,j}}AR(P')$.
        \item[b)] $d(C')/d(I') > AR(Q')$.\\
        Let $(C'',I'')$ be pair of circles such that $C''$ is the MCC and $I''$ the MIC of $Q'$. During our computation, we consider $(C'',I'')$ and find the same path $q'$ in $G_{i,j}^{(C'',I'')}$. Thus, $W(C'',I'')\leq \max\{w'(k',l'),d(C'')/d(I'')\}=\max\{w'(k',l'),AR(Q')\}$. Since both $w'(k',l')< d(C')/d(I')$ and $AR(Q')<d(C')/d(I')$, it follows that $W(C'',I'') < W(C',I') = w'(i,j)$. This contradicts our assumption that $w'(i,j)$ is minimal. 
    \end{description}
\end{description}
With this, we showed that $w(i,j)=\max_{P'\in Z_{i,j}}AR(P')$ holds.
\end{proof}

\begin{restatable}{lemma}{optimalpartition}\label{lem:optimalpartition}
The computed partition $Z_{i,j}$ is an optimal partition of $P_{i,j}$, meaning that the largest aspect ratio of any subpolygon is minimized.
\end{restatable}
\begin{proof}
We show this by induction over $j-i$. For $j-i=1$, we consider an edge $(i,j)=(i,i+1)$. Thus, $Z_{i,j}$ is the empty set.
Our induction hypothesis is that the computed partition $Z_{i,j}$ is an optimal partition of $P_{i,j}$ for all $(i,j)$ such that $j-i \leq s$. Now, we consider an edge $(i,j)$ such that $j-i=s+1$. Let $Z_{i,j}^*$ be an optimal partition of $P_{i,j}$. Hence, $\max_{P'\in Z_{i,j}^*}AR(P')$ is minimal over all partitions. Let $Q^*$ be the polygon in $Z^*_{i,j}$ adjacent to $(i,j)$ and let $q^*$ be the path induced by $Q^*$. The path $q^*$ is contained in $G_{i,j}^{(C^*,I^*)}$ for $(C^*,I^*)=(\text{MCC},\text{MIC})$ of $Q^*$. Thus, $q^*$ is a candidate for the min-max path computed by our algorithm and we have $W(C^*,I^*)=\max\{\max_{(k,l)\in q^*}w'(k,l),AR(Q^*)\}$. Let $Z_{i,j}=\bigcup_{(k,l)\in q'}Z_{k,l}\cup Q$ be the partition computed by our algorithm and $(C',I')$ the corresponding circles. We have:
\begin{align*}
W(C',I') &\leq W(C^*,I^*) \\  
\Rightarrow  \max\{\max_{(k,l)\in q'}w'(k,l),d(C')/d(I')\} &\leq \max\{\max_{(k,l)\in q^*}w'(k,l),AR(Q^*)\} \\
\underset{Lemma~\ref{lem:weight}}{\Rightarrow} \max_{P'\in Z_{i,j}}AR(P') &\leq \max_{P'\in Z_{i,j}^*}AR(P').
\end{align*}
Since the maximal aspect ratio in $Z_{i,j}$ is not larger that the maximal aspect ration in the optimal partition, it follows that the computed partition is optimal.    
\end{proof}

\begin{restatable}{theorem}{minmaxfat}\label{theo:minmaxfat}
For a simple polygon $P$, the min-fat partition problem using disk-fatness can be solved in time $\mathcal{O}(n^3m^5\log n)$ with $n$ being the number of vertices of $P$ and $m$ being the number of edges in the visibility graph $G(P)$.
\end{restatable}
\begin{proof}
First, we have to compute the visibility graph of $P$ which takes $\mathcal{O}(n+m)$ time. For every edge $(i,j)$ in the visibility graph, we determine an optimal partition $Z_{i,j}$ by computing the optimal weight $w(i,j)$. For each $(i,j)$, we consider pairs of circles $(C,I)$. There are $\mathcal{O}(n^3)$ circles $C$ and $\mathcal{O}(m^3)$ circles $I$ to consider. Computing the optimal path in $G_{i,j}^{(C,I)}$ can be done by an adjusted version of Dijkstra's Algorithm and therefore takes $\mathcal{O}(m\log n)$ time. Thus, the overall runtime of the algorithm is $\mathcal{O}(n^3m^5\log n)$.
\end{proof}

\section{Minimum \texorpdfstring{$\alpha$}{alpha}-fat decomposition problems for polygons with holes}\label{sec:minimum alpha fat decomposition}
In this section, we consider the minimum $\alpha$-fat decomposition problems on polygons with holes. We denote the corresponding decision problem by \emph{decide $\alpha$-fat partition} (resp. \emph{covering}) \emph{problem}. That is, deciding whether there exists an $\alpha$-fat partition (resp. covering) with a given number of components. We show NP-hardness with a reduction from \emph{planar 3,4-SAT}. For a boolean formula $\phi$, let $G(\phi)=(V,E)$ be the graph with $V=U\cup C$ and $E=\{(u,c)\,\vert\, u\in U, c\in C, \text{$u$ or $\Bar{u}$ is a literal in $c$}\}$, where $U$ are the literals and $C$ the clauses. Planar 3,4-SAT is the problem of deciding if $\phi$ is satisfiable under the following three restrictions:
In conjunctive normal form, $\phi$ has exactly 3 literals per clause,
   each literal appears in at most 4 clauses,
   and the graph $G(\phi)$ is planar.
We construct a polygon representing the graph $G(\phi)$ that has an $\alpha$-fat partition of size $k$ if and only if $\phi$ is satisfiable. The value of $\alpha$ is fixed and $k$ is determined during this construction.\par The related result for $\alpha$-smallness was proven by Worman~\cite{worman2003decomposing}. In their construction only (orthogonal) subpolygons below a certain height and width are feasible, as only the side length of the enclosing square is bounded. However, with $\alpha$-fatness it is the aspect ratio that is bounded. Thus, the challenge is to find suitable subpolygons and a suitable $\alpha$ such that all necessary components can be constructed. The constructions differ depending on the definition of fatness that is applied.

\subsection{Reduction for square-fatness}\label{sec:square reduction}
The construction consists of three different polygon components: variable, wire and clause polygons. We set $\alpha=1.2$. All polygon components are orthogonal and at most $5$ units wide, thus the side length of the enclosing squares cannot exceed $6$. The \emph{variable polygon} is shown in Figure~\ref{fig:variable}. It has four terminals at which wires can be connected. This polygon can be minimally partitioned into 8 $\alpha$-fat subpolygons in two ways (see Fig.~\ref{fig:variable_TF}). These partitions represent the \texttt{True} and \texttt{False} assignment, which is carried to the corresponding clauses. \par 
\begin{figure}[b]
	
	\centering
	\begin{subfigure}[c]{0.47\textwidth}
		\includegraphics[width=\textwidth]{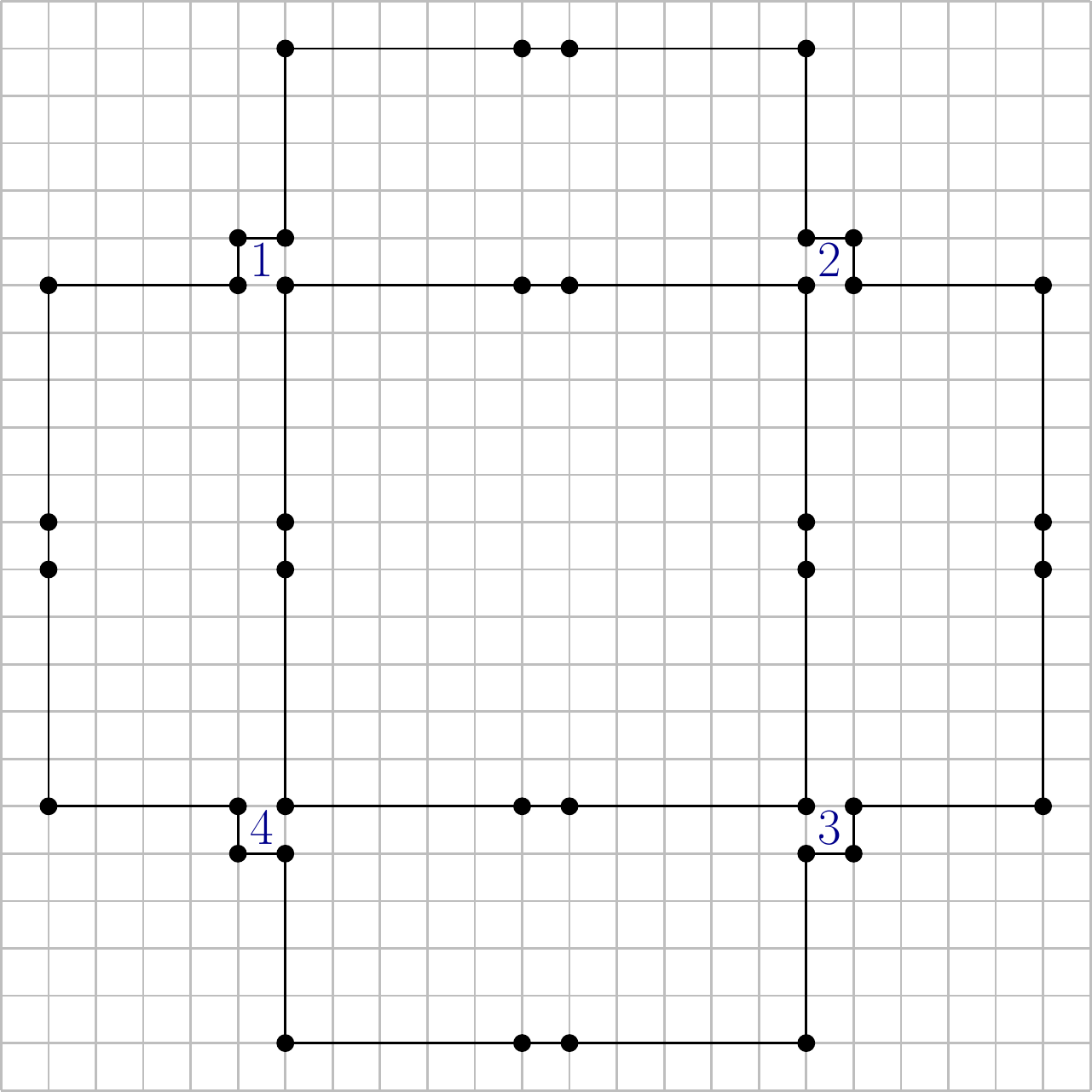}
		\caption{}
		\label{fig:variable}
	\end{subfigure}\qquad
	\begin{subfigure}[c]{0.47\textwidth}
		\includegraphics[width=0.93\textwidth]{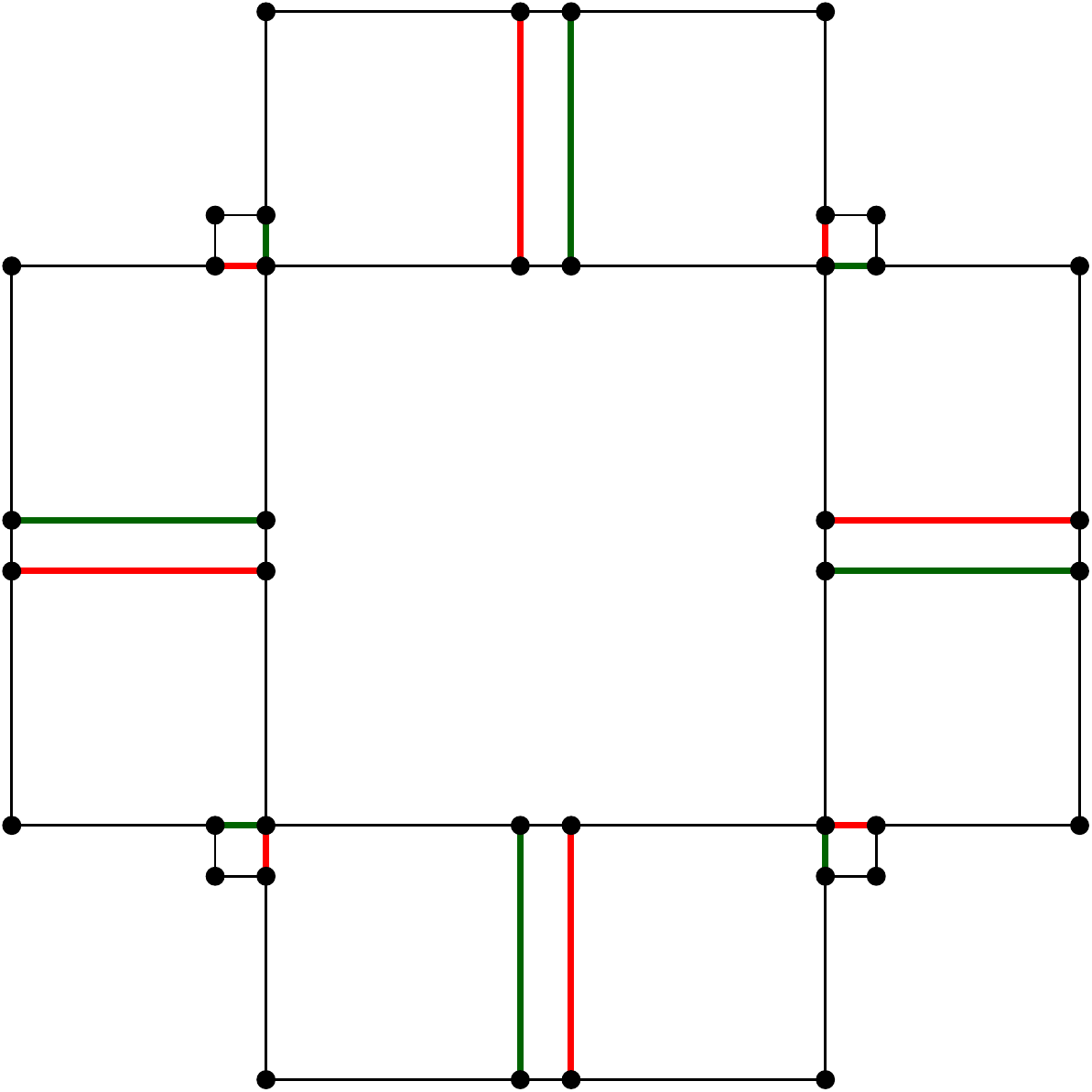}
		\caption{}
		\label{fig:variable_TF}
	\end{subfigure}
	\caption{The variable polygon with four terminals indicated by $1,2,3,4$ in blue (a) and its minimal 1.2-fat partitions representing the \texttt{True} (green) and \texttt{False} (red) assignment (b). }
	\label{fig:variablefig}
\end{figure}

\begin{figure}[t]
	
	\centering
	\begin{subfigure}[c]{0.28\textwidth}
		\includegraphics[height=1.5cm]{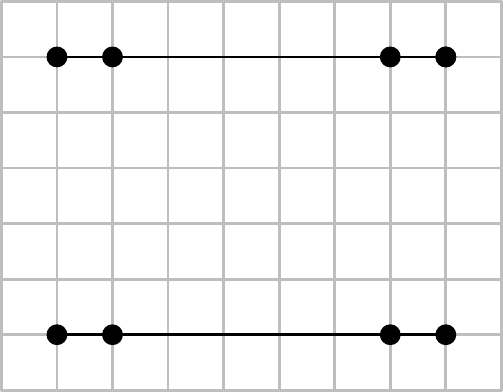}
		\caption{}
		\label{fig:wire}
	\end{subfigure}
	\begin{subfigure}[c]{0.28\textwidth}
		\includegraphics[height=1.3cm]{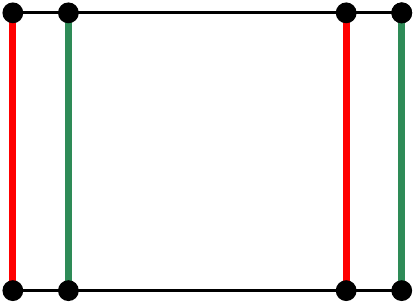}
		\caption{}
		\label{fig:wire_TF}
	\end{subfigure}
	\begin{subfigure}[c]{0.35\textwidth}
		\includegraphics[height=1.5cm]{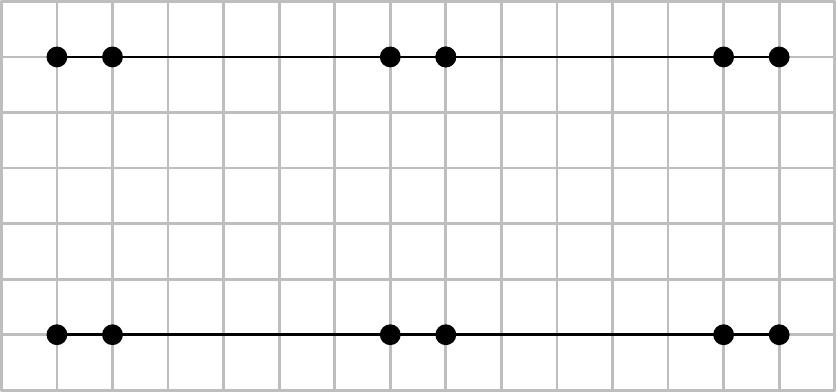}
		\caption{}
		\label{fig:2wire}
	\end{subfigure}
	\caption{A single wire polygon (a), its two partitions that represent the \texttt{True} (green) or \texttt{False} (red) assignment (b), and two connected wire polygons (c).}
	\label{fig:wirefig}
\end{figure}
Each wire consists of a set of individual \emph{wire polygons} that are connected with each other (as depicted in Figure~\ref{fig:wirefig}) to carry the variable assignment to the clause polygon. The wires can be attached at the terminals in two possible orientations (see Fig.~\ref{fig:attachfig}) depending on whether the variable appears in the clause negated or unnegated. If the wire is connected in the unnegated orientation and the variable is set to \texttt{True}, the green polygon in Figure~\ref{fig:attach_T} can cover the top part of the wire as well, but this is not the case if the variable is set to \texttt{False}. The reverse is true if the wire is connected in the negated orientation. If a wire is partitioned in this way (unnegated position and \texttt{True} assignment or negated position and \texttt{False} assignment), we say that it carries \texttt{True}. \par

\begin{figure}[t]
	
	\centering
	\begin{subfigure}[t]{0.45\textwidth}
		\includegraphics[width=\textwidth]{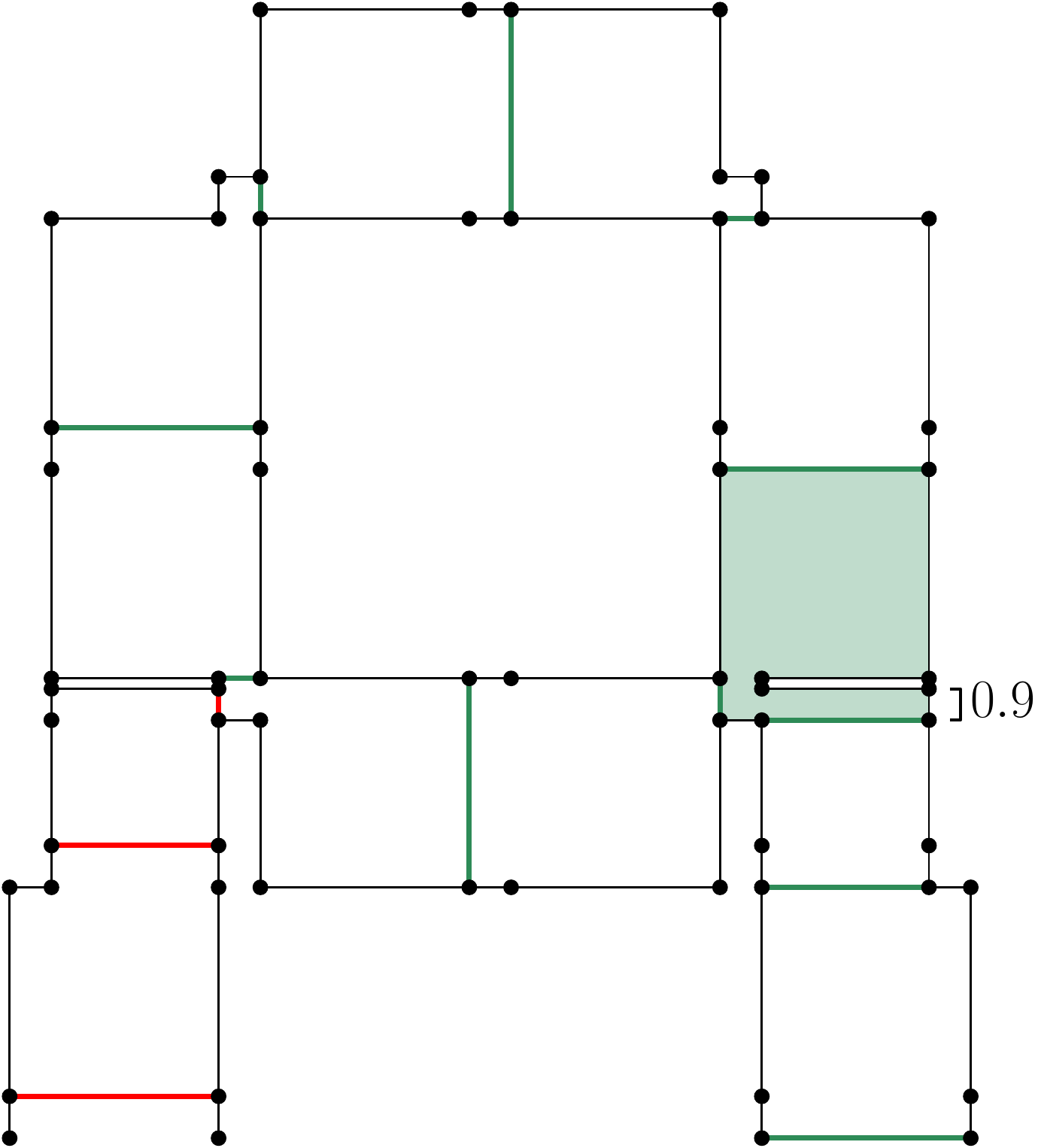}
		\caption{\texttt{True} assignment of variable.}
		\label{fig:attach_T}
	\end{subfigure}\qquad
	\begin{subfigure}[t]{0.45\textwidth}
		\includegraphics[width=\textwidth]{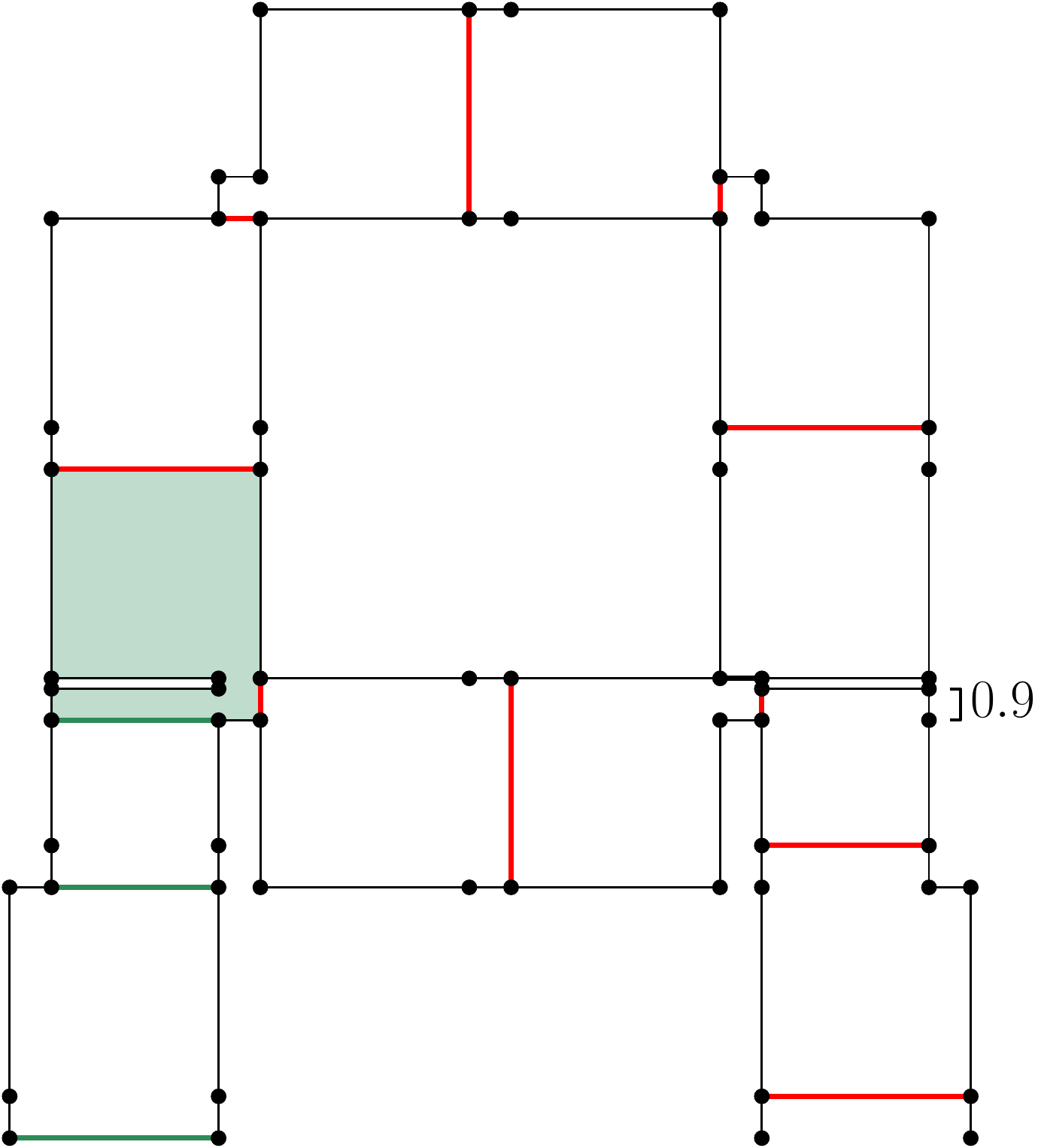}
		\caption{\texttt{False} assignment of variable.}
		\label{fig:attach_F}
	\end{subfigure}
	\caption{Attaching a wire to the variable in unnegated orientation (on bottom right terminal) and in negated orientation (on bottom left terminal) which switches the \texttt{True}/\texttt{False} value.}
	\label{fig:attachfig}
\end{figure}

The variable assignment is carried to the \emph{clause polygon} (see Fig.~\ref{fig:clausefig}). For each of the three variables contained in the clause, this polygon has one terminal, where the wires will be attached. Depending on the values these wires carry, a different number of polygons is needed to partition the clause polygon into 1.2-fat components. If some wire carries \texttt{True} (see Fig.~\ref{fig:clause_TFF} and \ref{fig:clause_TTT}), the tip of the connected terminal (gray) is already covered and center part of the clause polygon (dark green) can be covered as well. If more than one wire carries \texttt{True} either one of the corresponding polygons (light green) can be used to cover the center. In either case, the partition requires exactly four polygons.
If all wires carry \texttt{False} (see Fig.~\ref{fig:clause_FFF}), the center of the clause polygon depicted in red cannot be covered by any of the bigger polygons (in light red). Thus additional fifth polygon is needed to cover this area. As a $1.2$-fat polygon is required, the 1x1 square is chosen.\par
\begin{figure}[h] 
	
	\centering
	\includegraphics[width=5cm]{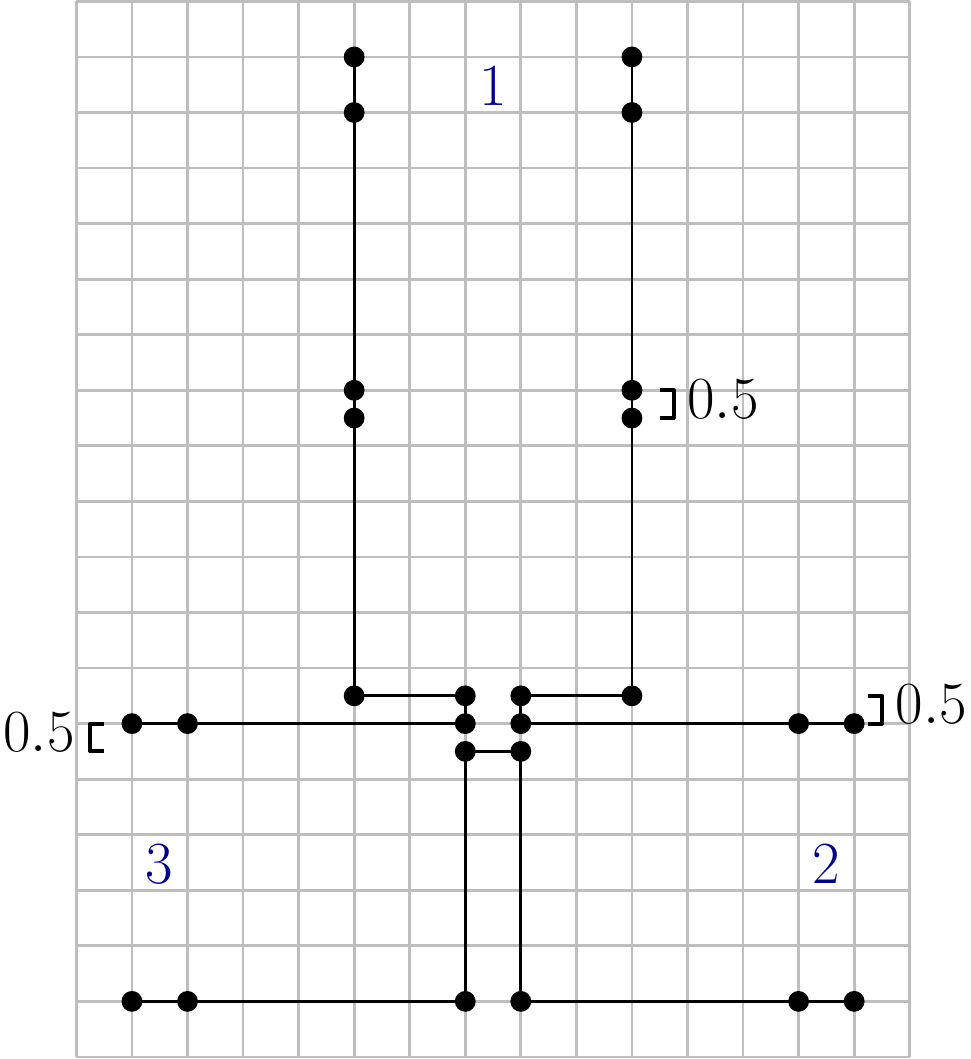}
	\caption{The clause polygon with three terminals indicated by $1,2,3$ in blue.}
	\label{fig:clausefig}
\end{figure}
\begin{figure}[h]
	
	\centering
	\begin{subfigure}[t]{0.3\textwidth}
	    \centering
		\includegraphics[height=5cm]{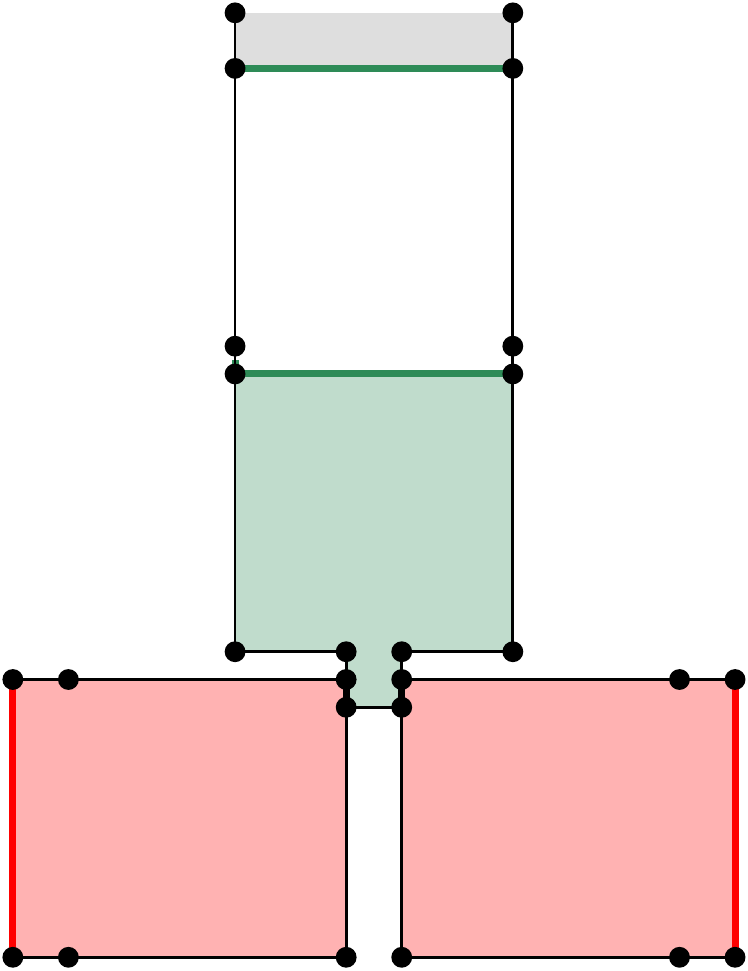}
		\caption{\texttt{True}, \texttt{False}, \texttt{False}}
		\label{fig:clause_TFF}
	\end{subfigure}\hfill
	\begin{subfigure}[t]{0.3\textwidth}
	    \centering
		\includegraphics[height=5cm]{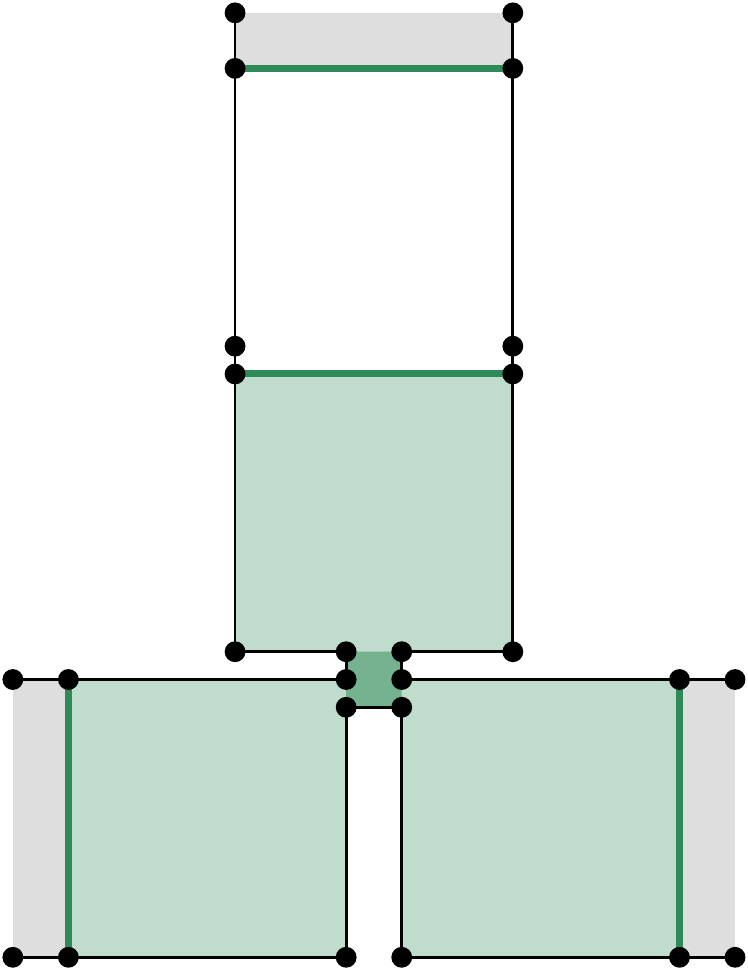}
		\caption{\texttt{True}, \texttt{True}, \texttt{True}}
		\label{fig:clause_TTT}
	\end{subfigure}\hfill
	\begin{subfigure}[t]{0.3\textwidth}
	    \centering
		\includegraphics[height=5cm]{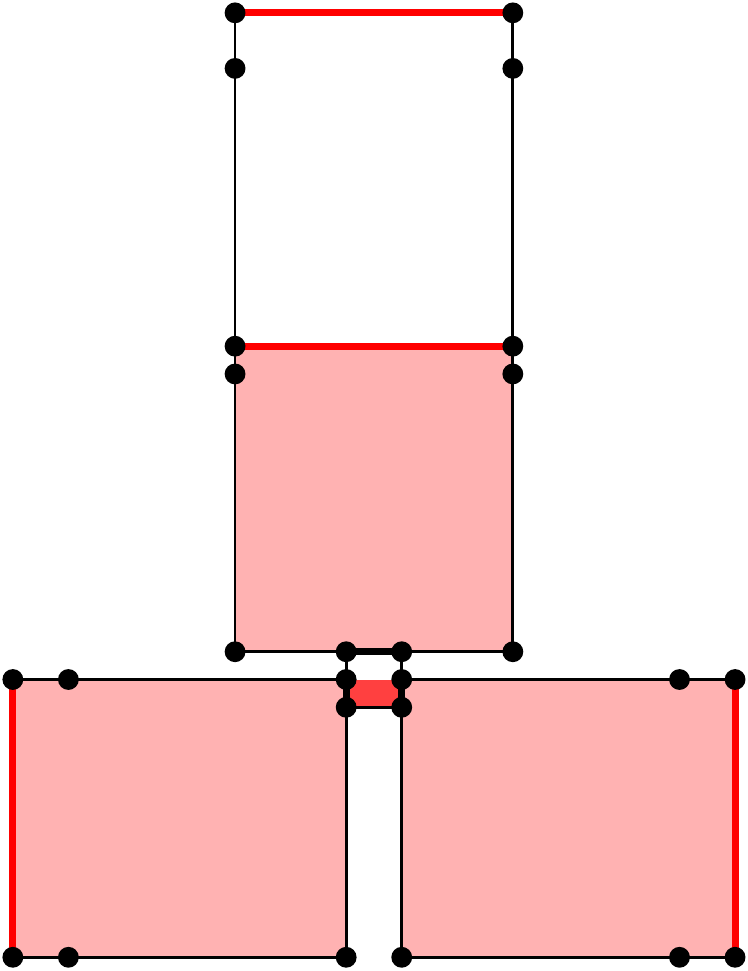}
		\caption{\texttt{False}, \texttt{False}, \texttt{False}}
		\label{fig:clause_FFF}
	\end{subfigure}
	\caption{Partition of the clause polygon depending on different assignments that are transmitted by the wires (\texttt{True} green edges, \texttt{False} red edges).}
	\label{fig:clausesetting}
\end{figure}
The whole polygon representing $G(\phi)$ is constructed based on a \emph{planar orthogonal grid drawing} of $G(\phi)$. That is a planar embedding of the graph such that every vertex is located at an integer grid point, the edges are non-overlapping, and every edge is a chain of orthogonal lines that bend at integer grid points. A schematic example for the placement of variable and clause polygons on the vertices of the drawing can be seen in Figure~\ref{fig:placementgrid}. To construct the edges, the wires have to be bend, shifted or offset. This is achieved by the constructions presented in Figure~\ref{fig:constructwire}. The drawing of $G(\phi)$ is scaled to accommodate the size of the variable and clause polygons as well as the needed adjustments of the wire polygons.\par

\begin{figure}[htb] 
	
	\centering
	\includegraphics[width=7cm]{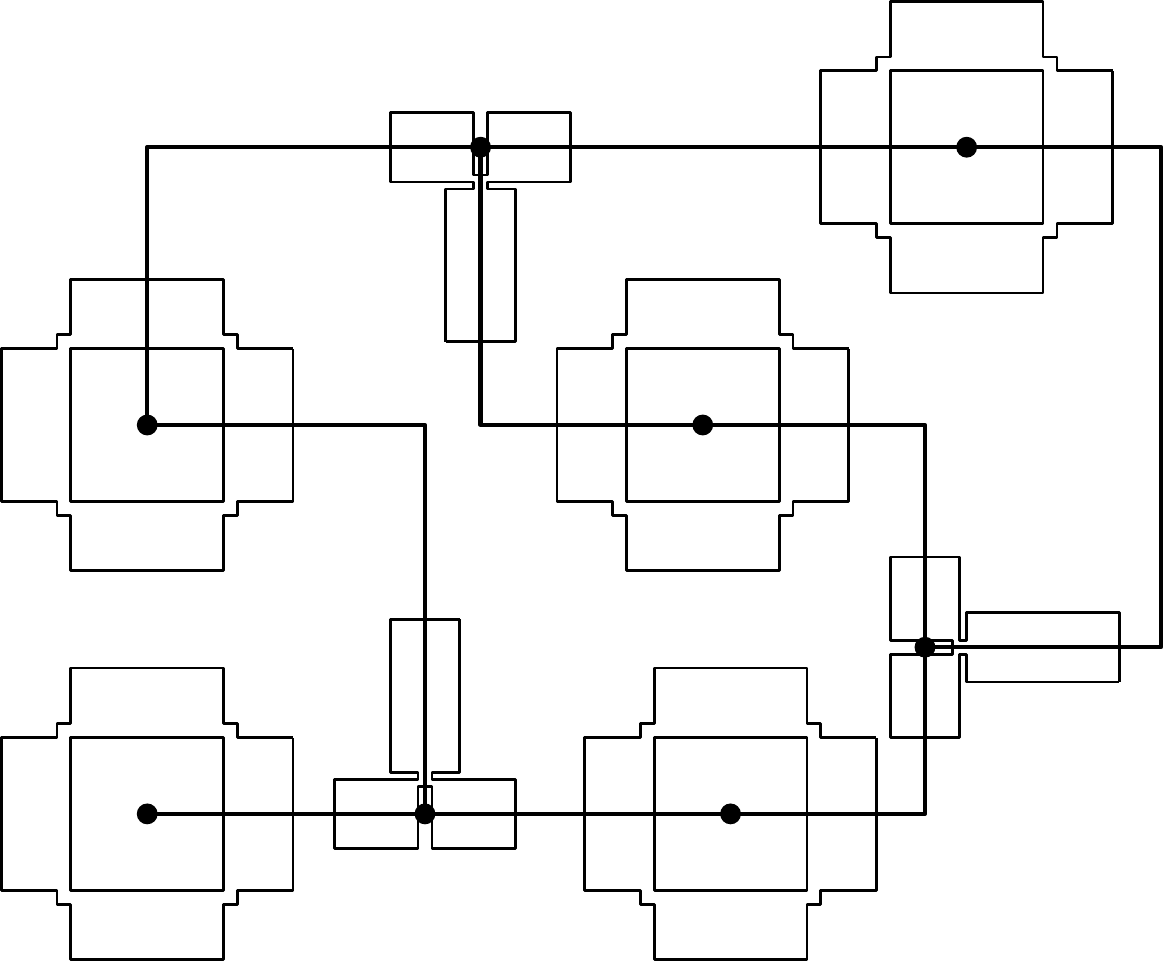}
	\caption{Placement of $5$ variable and $3$ clause polygons on a planar orthogonal grid drawing.}
	\label{fig:placementgrid}
\end{figure}
\clearpage
\begin{figure}[htb]
	
	\centering
	\begin{subfigure}[t]{0.32\textwidth}
		\includegraphics[width=\textwidth]{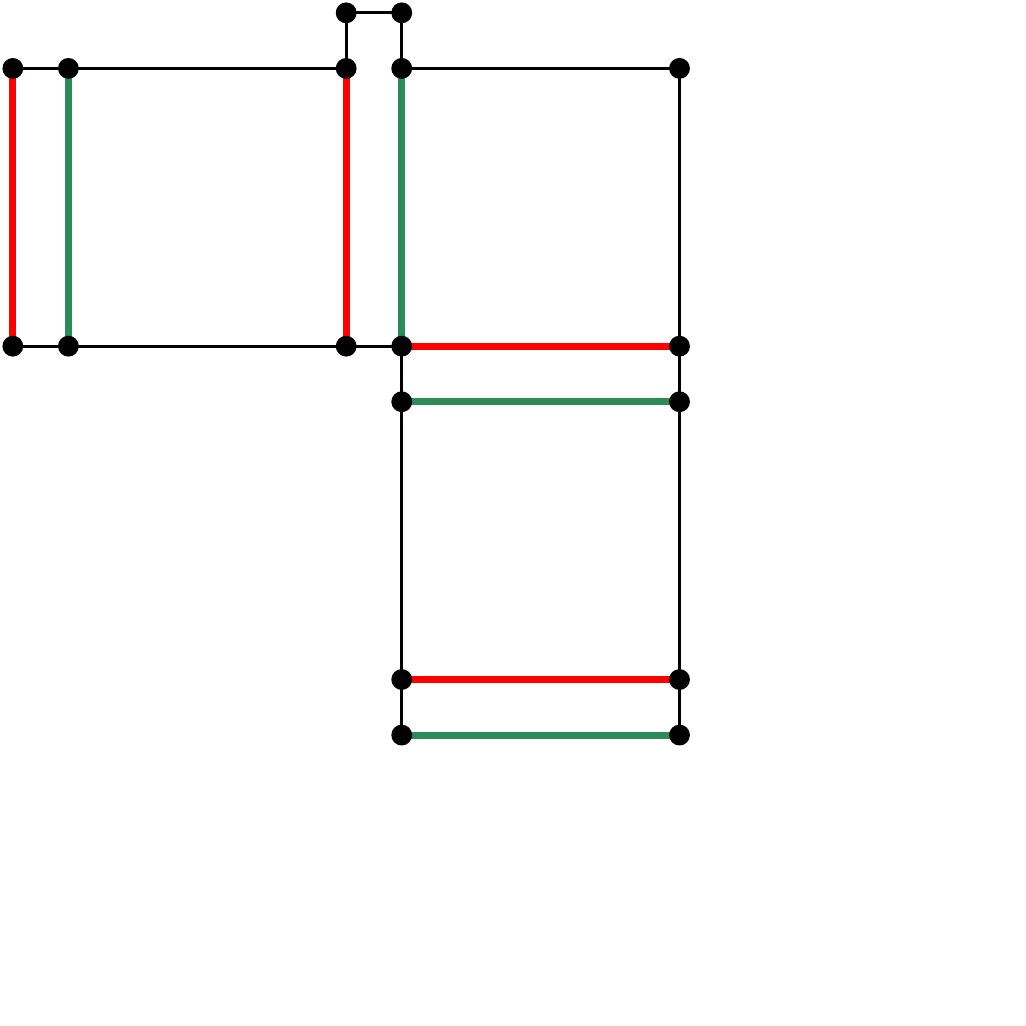}
		\caption{}
		\label{fig:bend}
	\end{subfigure}
	\begin{subfigure}[t]{0.32\textwidth}
		\includegraphics[width=\textwidth]{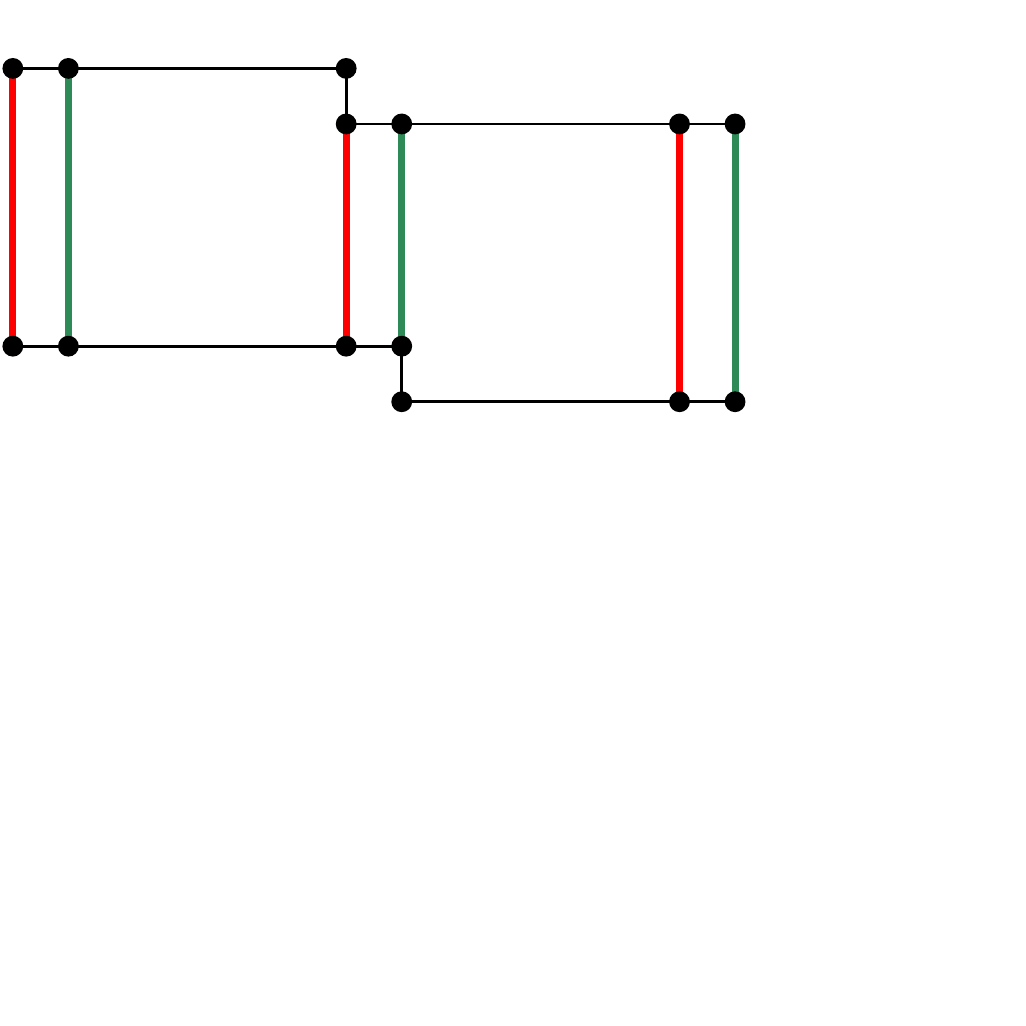}
		\caption{}
		\label{fig:shift}
	\end{subfigure}
	\begin{subfigure}[t]{0.32\textwidth}
		\includegraphics[width=\textwidth]{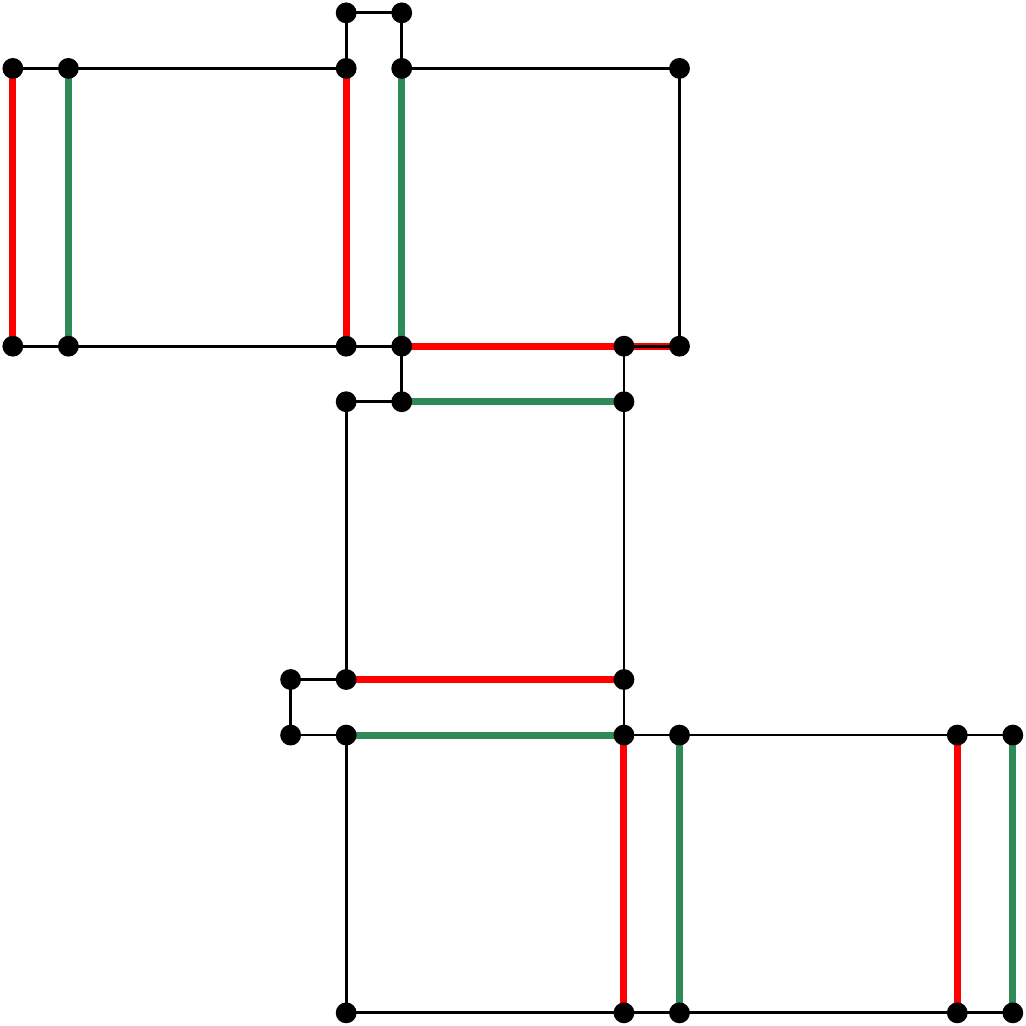}
		\caption{}
		\label{fig:offset}
	\end{subfigure}
	\caption{Bending\,(a), shifting\,(b) and offsetting\,(c) a wire that carries \texttt{True} (green edges) or \texttt{False} (red edges).}
	\label{fig:constructwire}
\end{figure}

As we consider the decision problem, the number $k$ of allowed subpolygons is fixed. We have $k=8v+4c+w$ where $v$ is the number of variables, $c$ the number of clauses and $w$ the number of wire polygons needed in the construction. Bending, shifting and offsetting a wire counts as 3, 2, and 5 wire polygons respectively. 
Note that we can find a minimum $1.2$-fat covering with the same number of components as the minimum partition and that the constructed polygon is orthogonal. This means that the following theorem remains true for the corresponding covering problem and also for orthogonal polygons with holes.
\begin{restatable}{theorem}{alphafat}
The decide $\alpha$-fat partition problem is NP-complete for polygons with holes if square-fatness is applied.
\end{restatable}
\begin{proof}
Given a set of $k$ $\alpha$-fat polygons, we can verify in polynomial time if this is a partition of our polygon $P$. Hence, the problem is in NP.

Let $\phi$ be an instance of planar 3,4-SAT and $P$ be the polygon representing $G(\phi)$ constructed as described in Section~\ref{sec:minimum alpha fat decomposition}. We show that $P$ has a minimum 1.2-fat partition of a certain size $k$ if and only if $\phi$ is satisfiable.

Assume that $\phi$ can be satisfied with a truth assignment $T$. Then, the polygon can be minimally partitioned into $k=8v + w + 4 c$ components. The values $v$ and $c$ are the numbers of variables and clauses respectively. The value $w$ corresponds to the number of wire polygons needed to construct $G(\phi)$ with $P$. For each of the $v$ variables, we partition the variable polygon into 8 components according to $T$. The wires are partitioned into $w$ subpolygons. Depending on the assignment of the variables the tip of the terminal (the part where the wire and clause polygon overlap) might be uncovered. Since the assignment $T$ satisfies $\phi$, we know that for each clause there is at least one wire carrying \texttt{True} such that the tip of the terminal is already covered and the clause polygon can be partitioned into 4 components. 

Assume that $P$ has a minimum 1.2-fat partition $Z$ of size $k = 8v + w + 4 c$ as above. There are two ways how to minimally partition a variable polygon into 1.2-fat polygons, which both require 8 components. To partition a sequence of $l$ wire polygons, $l$ subpolygons are needed. If the tip of the wire at the variable terminal is already covered, the wire carries \texttt{True} and the tip at the clause terminal will be covered as well. If we use $w$ polygons to partition the wires, 4$c$ polygons remain for the partition of the clauses. Since each of the $c$ clause polygons can only be partitioned into 4 components if the tip of one of the terminals is already covered by another polygon -- which is only the case if the corresponding wire carries \texttt{True} -- we know that each clause is satisfied. Hence, $\phi$ is satisfied. 
\end{proof}

\subsection{Reduction for disk-fatness}\label{sec:disk reduction}
The construction presented in Section~\ref{sec:square reduction} does not work if disk-fatness is applied. For all previously feasible (with $AR_{square}\leq 1.2=6/5$) subpolygons to remain feasible, we would have to set $\alpha_{disk}=6/5\sqrt{2}$. However, there exist other subpolygons in the construction with $AR_{disk}\leq 6/5\sqrt{2}$ but $AR_{square}>6/5$ -- specifically at the wire bends and inside the clause polygon. These would then become feasible and cause an inconsistent transmission of \texttt{True}/\texttt{False} values. We adjust the construction such that the aspect ratio $AR_{disk}$ of all subpolygons that are supposed to be feasible is at most $\alpha=\sqrt{61}/5 \approx 1.56$.\par
For variable polygons and their attachment to the wires nothing changes. Wires can be shifted in the same way, but bending and offsetting has to be adjusted. 
We have to remove the 1x1-square bulge at the corners because the corner polygon in the red partition in Figures~\ref{fig:bend} and~\ref{fig:offset} would no longer be feasible. The adjusted wires and their feasible partitions are shown in Figure~\ref{fig:constructwire_circles} and the respective corner polygons are depicted in Figure~\ref{fig:circle_wires}. The two polygons in Figure~\ref{fig:bend_F} and~\ref{fig:bend_T} have the correct aspect ratio of $AR_{disk}=\sqrt{61}/5$. The other polygon shown in~\ref{fig:offset_not} has a MCC of diameter $6\sqrt{2}$ and MIC of diameter $2(9-\sqrt{40})$. Therefore, the aspect ratio is larger than $\sqrt{61}/5=\alpha$, which means that it is not feasible in our construction. Note that this polygon, however, would be feasible regarding square-fatness with $\alpha=1.2$.\par
\begin{figure}[h]
	
	\centering
	\begin{subfigure}[b]{0.35\textwidth}
		\includegraphics[height= 3.5cm]{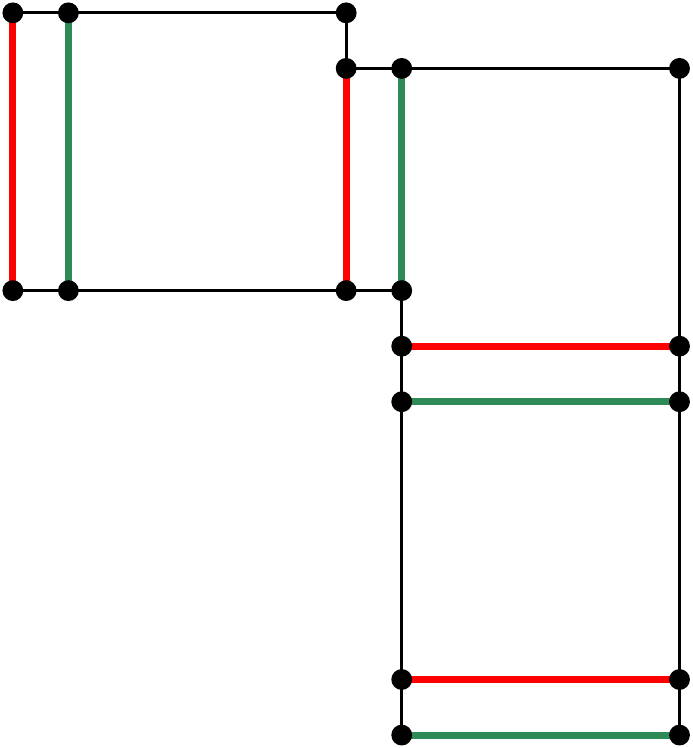}
		\caption{}
		\label{fig:bend_circle}
	\end{subfigure}\qquad
	\begin{subfigure}[b]{0.45\textwidth}
		\includegraphics[height=3.5cm]{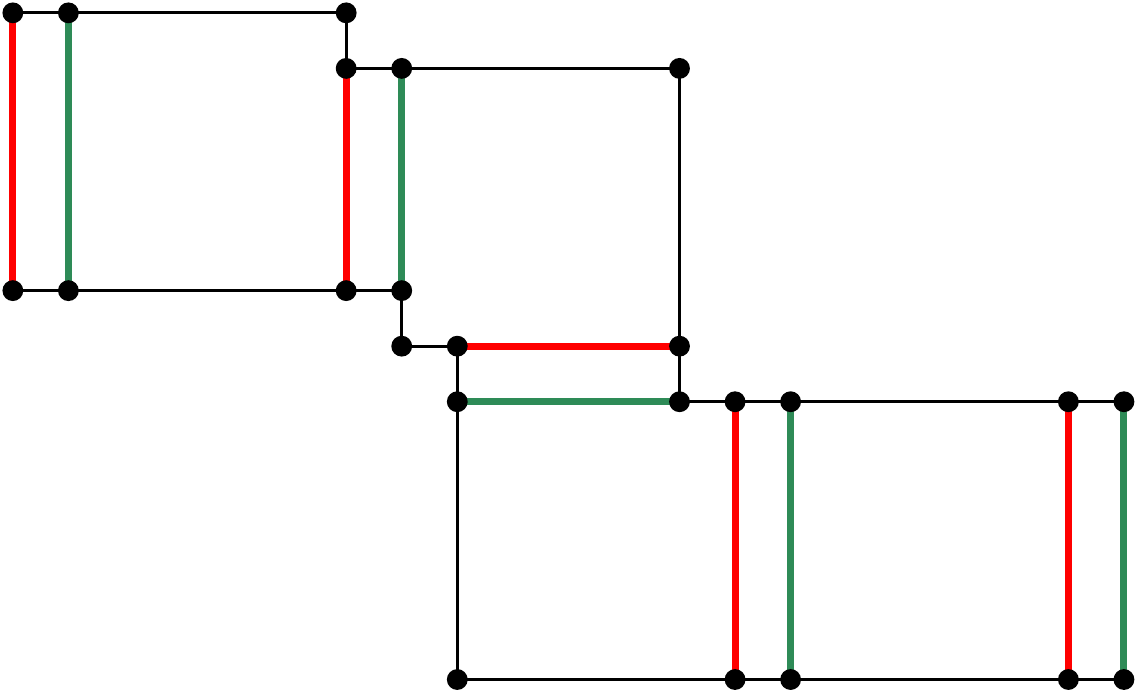}
		\caption{}
		\label{fig:offset_circle}
	\end{subfigure}
	\caption{Bending (a) and offsetting (b) a wire that carries \texttt{True} (green edges) or \texttt{False} (red edges) in the  disk-fatness construction.}
	\label{fig:constructwire_circles}
\end{figure}

\begin{figure}[h]
	
	\centering
	\begin{subfigure}[b]{0.3\textwidth}
		\includegraphics[width=\textwidth]{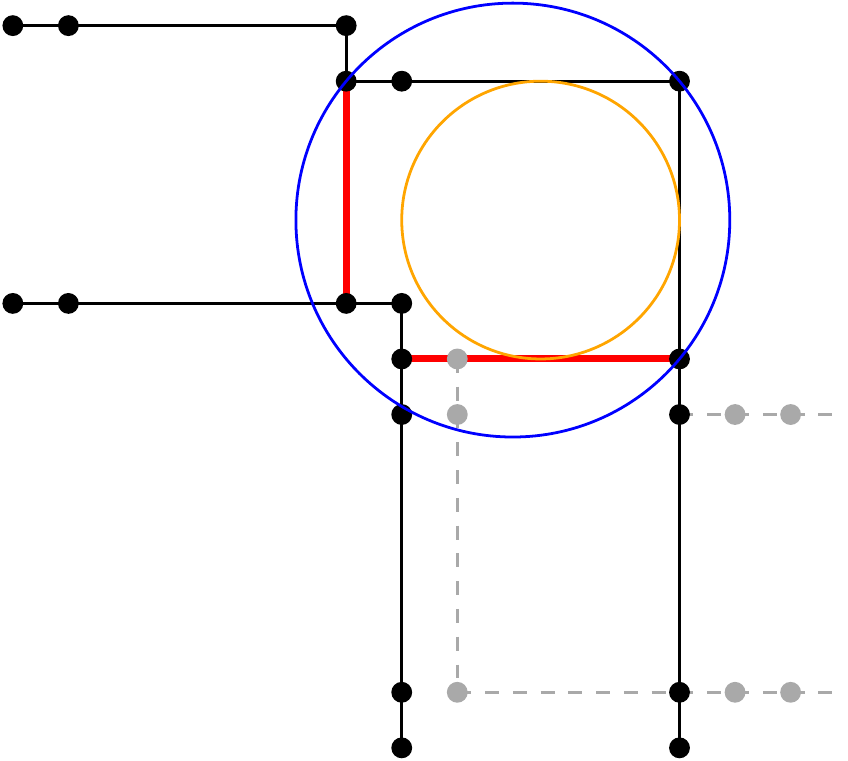}
		\caption{Corner polygon for \texttt{False}.}
		\label{fig:bend_F}
	\end{subfigure}\hfill
	\begin{subfigure}[b]{0.3\textwidth}
		\includegraphics[width=\textwidth]{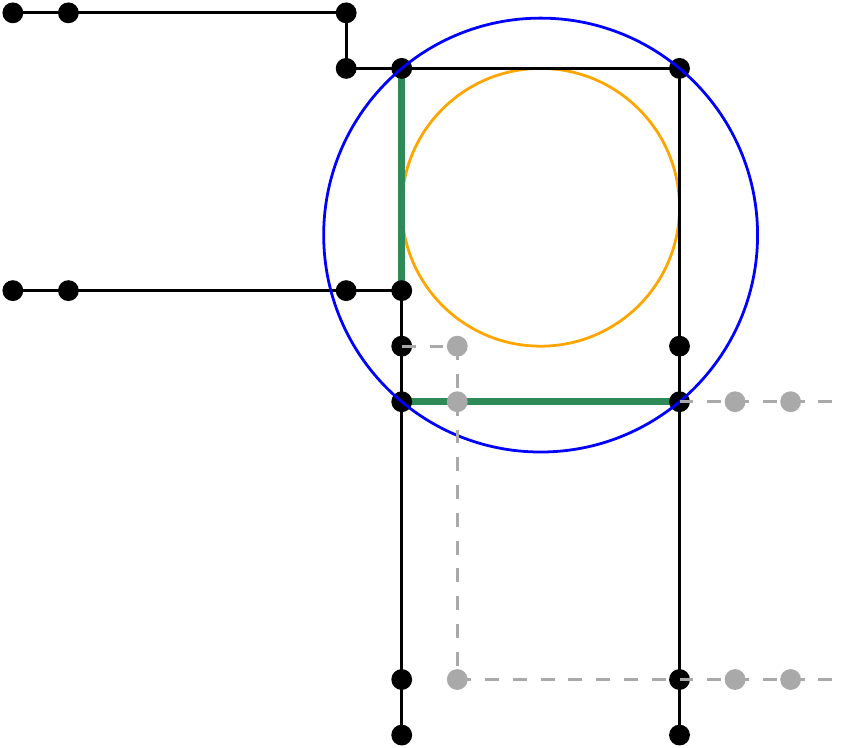}
		\caption{Corner polygon for \texttt{True}.}
		\label{fig:bend_T}
	\end{subfigure}\hfill
	\begin{subfigure}[b]{0.3\textwidth}
		\includegraphics[width=\textwidth]{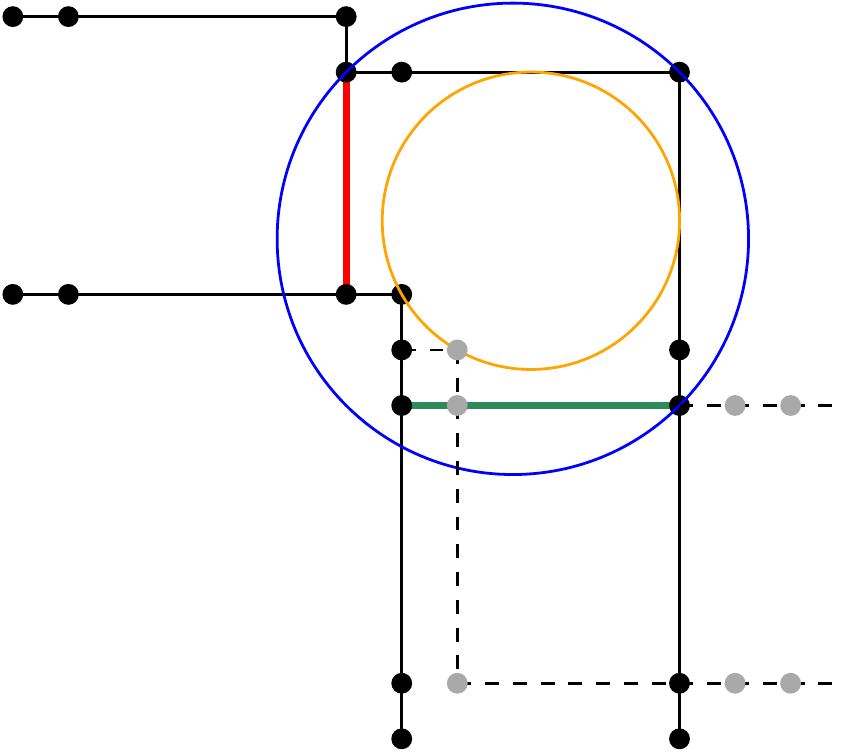}
		\caption{Not feasible corner polygon.}
		\label{fig:offset_not}
	\end{subfigure}
	\caption{Three possible corner polygons in a wire polygon that is bend (or offset in gray) with their corresponding MCC (blue) and MIC (orange). Only the polygons in (a) and (b) are $\alpha$-fat with the given $\alpha=\sqrt{61}/5$ and hence feasible.}
	\label{fig:circle_wires}
\end{figure}
\newpage
Furthermore, we have to slightly adjust the clause polygon. Figure~\ref{fig:clause_before} shows the previous construction for square-fatness. The illustrated polygon has $AR_{square} = 6.5/5$ and was not feasible before. However, it fulfills $AR_{disk} = \sqrt{61}/5$ (MCC and MIC are shown in blue and resp. orange) and thus it would cause inconsistent  \texttt{True}/\texttt{False} values.
By changing the placement of four boundary vertices (shown in blue in Fig.~\ref{fig:clause_circle}), we can ensure the correct transmission. Meaning, the clause polygon can be decomposed into four feasible $\alpha$-fat subpolygons exactly if at least one of the wires transmits \texttt{True}. 
Figure~\ref{fig:clause_with_circles} exemplifies the idea behind this construction. Let $P_F$ and $P_T$ be the two polygons depicted in Figure~\ref{fig:clausecircle_F} and~\ref{fig:clausecircle_T} respectively. The polygon $P_F$ has $AR_{disk}=\sqrt{61}/5=\alpha$ and $P_T$ has $AR_{disk}<\alpha$. Therefore, both subpolygons are feasible in our construction. The position of the four blue vertices (as marked in Fig.~\ref{fig:clause_circle}) was adjusted in such a way that the lower two are contained in the MCC of $P_T$ but not in the MCC of $P_F$. As the aspect ratio of $P_F$ is exactly $\alpha$, no larger polygon (coming from terminal 1) can be chosen to cover the area between these vertices and thereby the center area of the clause polygon. Thus, the center can only be covered if another wire carries \texttt{True} or an additional fifth polygon is included in the partition.\par 
\begin{figure}[h]
    \centering
	\begin{subfigure}[b]{0.45\textwidth}
	    \centering
		\includegraphics[width=5cm]{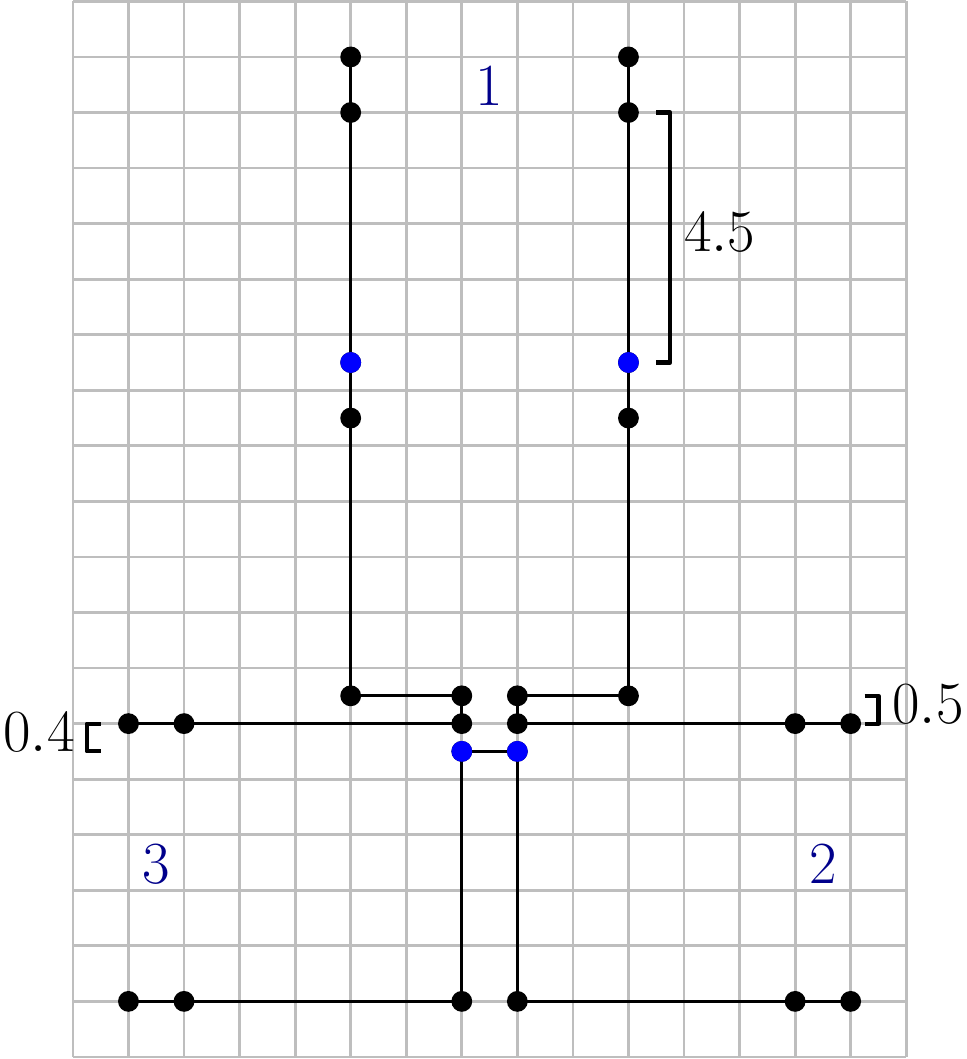}
		\caption{The clause polygon for disk-fatness.}
		\label{fig:clause_circle}
	\end{subfigure}\hfill
	\begin{subfigure}[b]{0.45\textwidth}
	    \centering
		\includegraphics[width=5cm]{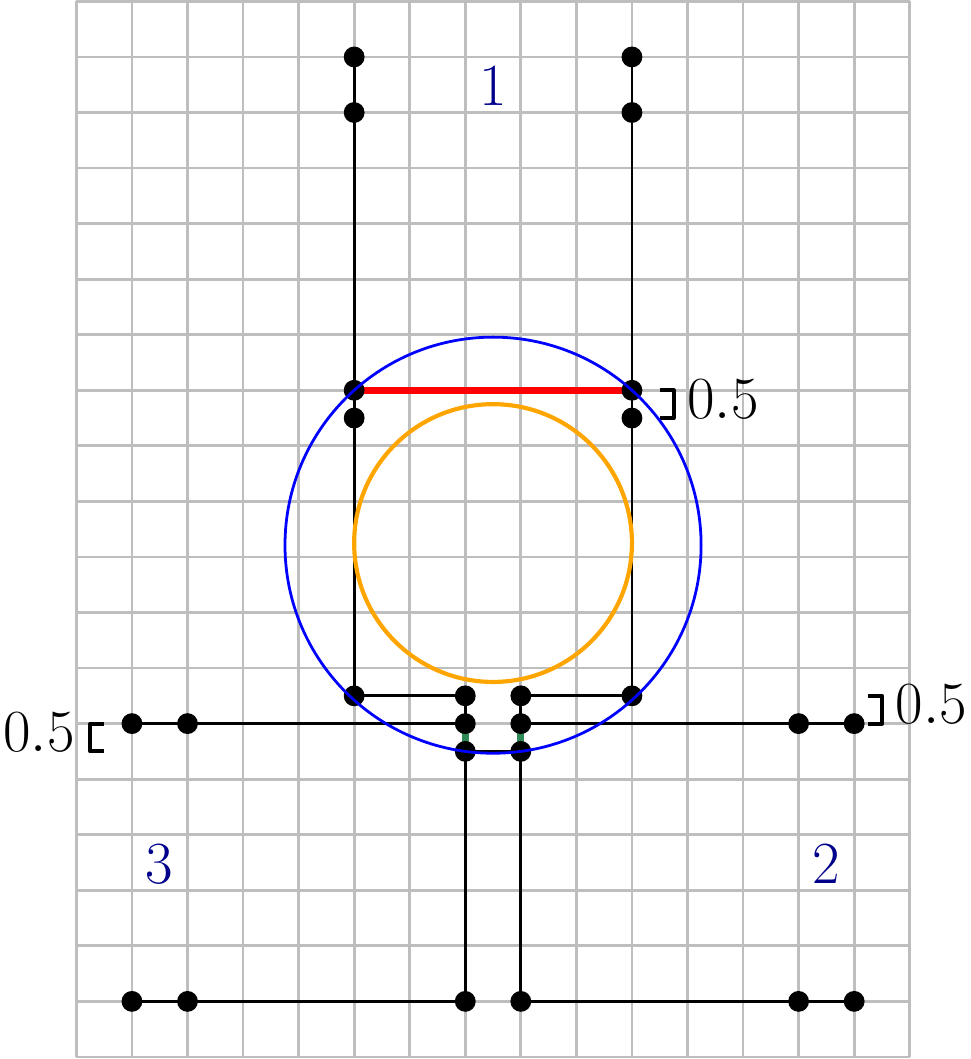}
		\caption{The clause polygon for square-fatness.}
		\label{fig:clause_before}
	\end{subfigure}
	\label{fig:clause_disk}
	\caption{The adjusted construction of the clause polygon for disk-fatness in comparison to the one for square-fatness.   }
\end{figure}	

\begin{figure}[h]
	
	\centering
	\begin{subfigure}[t]{0.3\textwidth}
	    \centering
		\includegraphics[height=5cm]{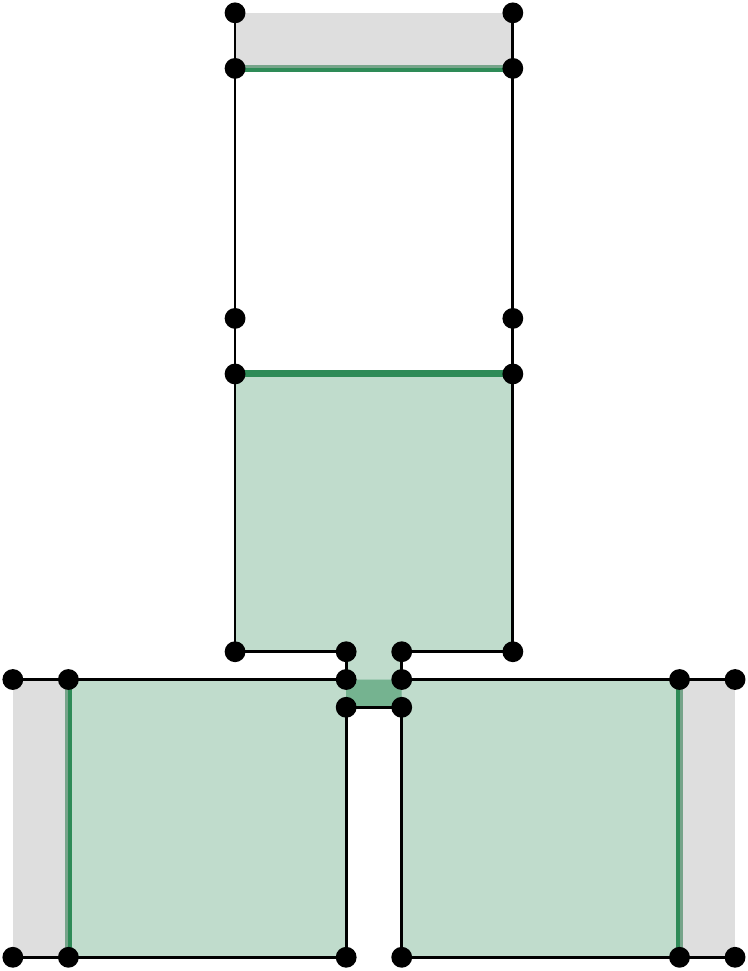}
		\caption{\texttt{True}, \texttt{True}, \texttt{True}}
		\label{fig:clause_TTT_circle}
	\end{subfigure}\hfill
	\begin{subfigure}[t]{0.3\textwidth}
	    \centering
		\includegraphics[height=5cm]{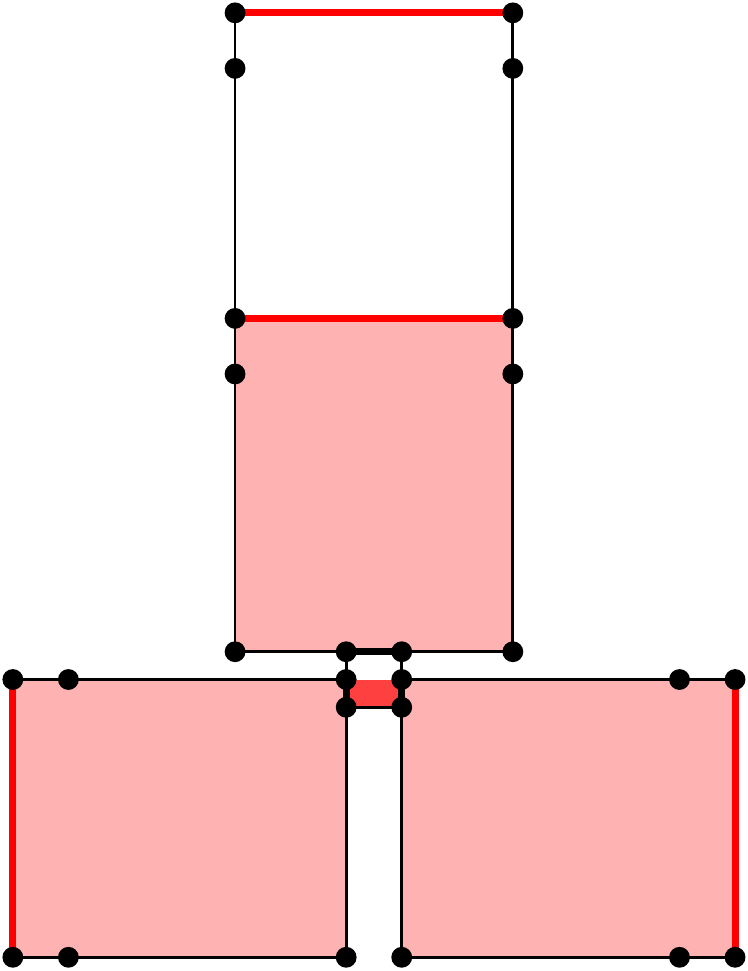}
		\caption{\texttt{False}, \texttt{False}, \texttt{False}}
		\label{fig:clause_FFF_circle}
	\end{subfigure}\hfill
	\begin{subfigure}[t]{0.3\textwidth}
	    \centering
		\includegraphics[height=5cm]{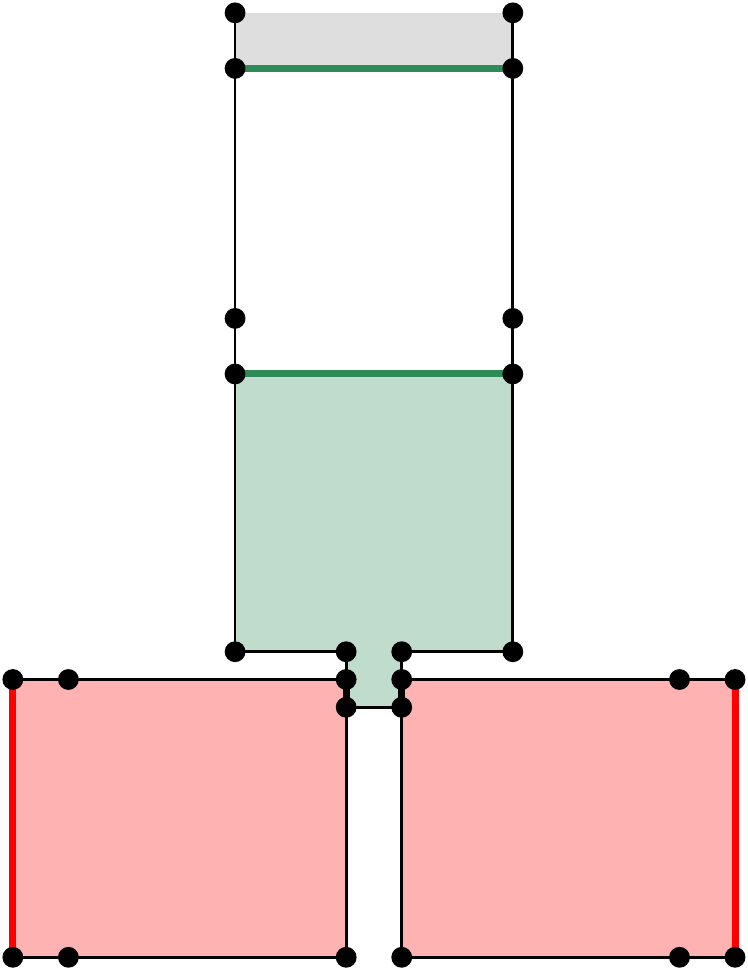}
		\caption{\texttt{True}, \texttt{False}, \texttt{False}}
		\label{fig:clause_TFF_circle}
	\end{subfigure}
	\caption{Partition of the clause polygon for disk-fatness depending on different assignments that are transmitted by the wires (\texttt{True} green edges, \texttt{False} red edges).}
		\label{fig:clausefig_circles}
\end{figure}

\begin{figure}[h]
	
	\centering
	\begin{subfigure}[b]{0.45\textwidth}
	\centering
		\includegraphics[width=4cm]{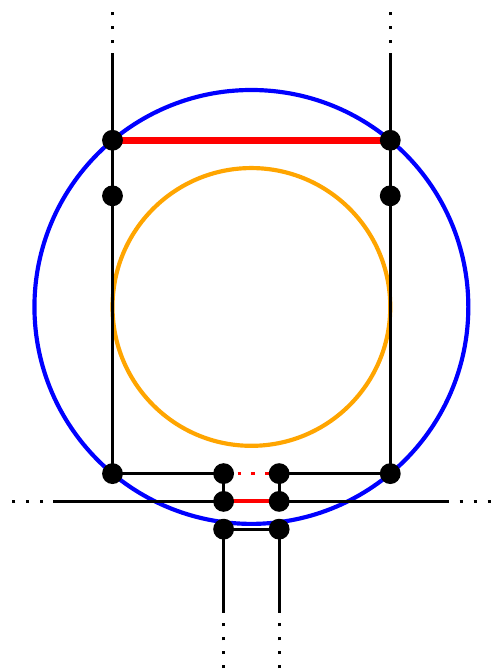}
		\caption{Subpolygon for \texttt{False}.}
		\label{fig:clausecircle_F}
	\end{subfigure}\qquad
	\begin{subfigure}[b]{0.45\textwidth}
	\centering
		\includegraphics[width=4cm]{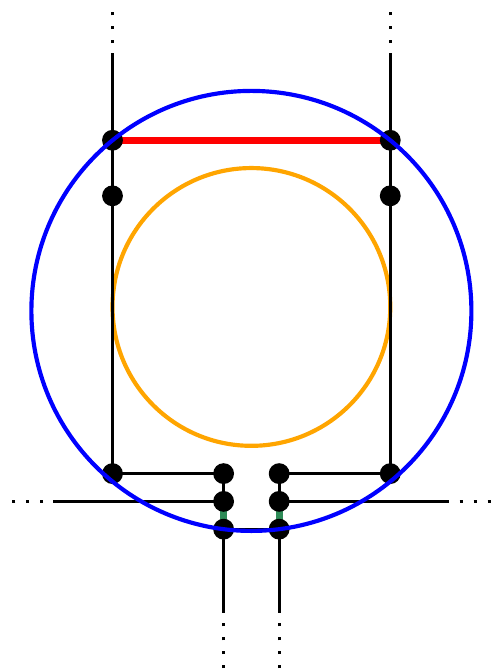}
		\caption{Subpolygon for \texttt{True}.}
		\label{fig:clausecircle_T}
	\end{subfigure}
	\caption{Subpolygons in the clause polygon for disk-fatness with their corresponding MCC (blue) and MIC (orange).}
	\label{fig:clause_with_circles}
\end{figure}	
	
Note that the construction resulting from these adjustments would not work for square-fatness. Again other subpolygons that interfere with the correct transmission of \texttt{True} and \texttt{False} values would become feasible. For example, the corner polygon in Fig.~\ref{fig:offset_not} would be feasible with the previous $\alpha$ for square-fatness. Using the adjusted construction, we can reduce from planar 3,4-SAT and thus prove the NP-completeness in the same way as before.
\begin{theorem}
Decide $\alpha$-fat partition problem is NP-complete for polygons with holes if disk-fatness is applied.
\end{theorem} 
As argued before, this result remains true for the covering problem and additionally for orthogonal polygons with holes.

\section{Conclusion}
We presented a polynomial-time algorithm for the min-fat partition problem for simple polygons. Furthermore, we proved that it is NP-complete to decide the $\alpha$-fat partition problem and covering problem for polygons with holes. Both results are true for disk-fatness and the latter also holds for square-fatness. For polygons with holes, the min-fat partition problem remains open. Furthermore, it is unclear whether the minimum $\alpha$-fat or min-fat covering problem is solvable for simple polygons. Moreover, there are no results for any related small and fat decomposition problems that allow Steiner points yet.


\bibliography{Literature}

\appendix
\section{Minimum \texorpdfstring{$\alpha$}{alpha}-small decomposition problem with disk-fatness}\label{app:disk reduction}
Worman showed that the minimum $\alpha$-small decomposition problems are NP-complete if the definition of square-smallness is used~\cite{worman2003decomposing}. In their construction, they used $\alpha_{square}=3$ and all feasible subpolygons can be enclosed with a 3x3 square. They claim that their construction holds for disk-smallness as well if $\alpha_{disk}=\sqrt{18}$, which is the diameter of the MCC of the 3x3 square. However, this is incorrect. Their feasible wire polygons are 1x3 rectangles, but the wires can also be partitioned in 1x4 rectangles. The latter were not feasible for $\alpha_{square}$, but they would become feasible with $\alpha_{disk}$ because the diameter of their MCC is only $\sqrt{17}$. Additionally, there are subpolygons inside the clause polygon, which have an MCC of diameter $\sqrt{17}$ but are not supposed to be feasible. For their result to be true, a different $\alpha_{disk}$ has to be chosen and the construction has to be adjusted.\par 
We choose $\alpha_{disk}=\sqrt{13}$. The variable polygon stays the same, as $\alpha_{disk}$ is exactly the diameter of the MCC of the subpolygons that supposed to be feasible. The basic wire polygon can also stay the same, but bending and offsetting a wire has to be adjusted (see Fig.~\ref{fig:constructwire_small}). This is done by again removing the bulges at the corner polygons. Note that all feasible subpolygons now fit inside a 2x3 rectangle and therefore also fulfill $\alpha_{disk}$. The clause polygon can be constructed as a slimmed down version of the clause polygon presented in Section~\ref{sec:minimum alpha fat decomposition} (see Fig.~\ref{fig:constructclause_small}). The center square that can only be covered by another polygon coming from the terminals if at least one of the connected wires carries \texttt{True}. Using the presented construction, we can prove that the problem is NP-complete with an analog reduction from planar 3,4-SAT. 
\begin{theorem}
Decide $\alpha$-small partition and covering problems are NP-complete for polygons with holes if disk-fatness is applied.
\end{theorem} 
Note that this adjusted construction does not work for square-smallness anymore.

\begin{figure}[h]
	
	\centering
	\begin{subfigure}[t]{0.45\textwidth}
		\includegraphics[scale=0.9]{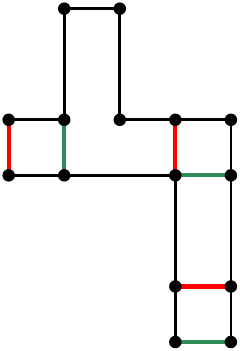}
		\caption{Bending a wire for square-smallness.}
		\label{fig:bend_small_old}
	\end{subfigure}\qquad
	\begin{subfigure}[t]{0.45\textwidth}
		\includegraphics[scale=0.9]{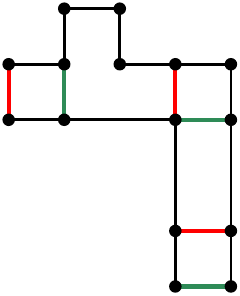}
		\caption{Bending a wire bend for disk-smallness.}
		\label{fig:bend_small_new}
	\end{subfigure}\medskip\\
	\begin{subfigure}[t]{0.45\textwidth}
		\includegraphics[scale=0.9]{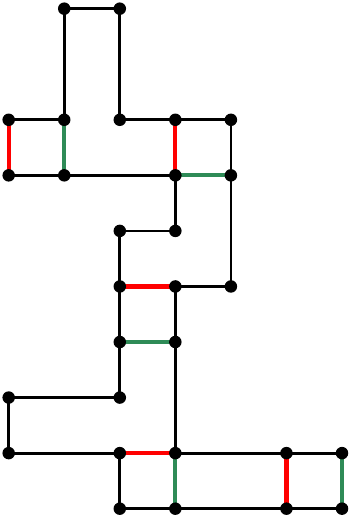}
		\caption{Offsetting a wire for square-smallness.}
		\label{fig:offset_small_old}
	\end{subfigure}
    \qquad
	\begin{subfigure}[t]{0.45\textwidth}
		\includegraphics[scale=0.9]{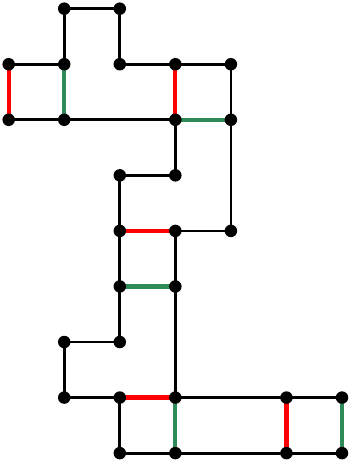}
		\caption{Offsetting a wire for disk-smallness.}
		\label{fig:offset_small_new}
	\end{subfigure}
	\caption{Comparison of wire construction for square-smallness (left) and disk-smallness (right). Bending (top) and offsetting (bottom) a wire that carries \texttt{True} (green edges) or \texttt{False} (red edges).}
	\label{fig:constructwire_small}
\end{figure}

\begin{figure}[h]
	
	\centering
	\begin{subfigure}[t]{0.45\textwidth}
		\includegraphics[scale=0.9]{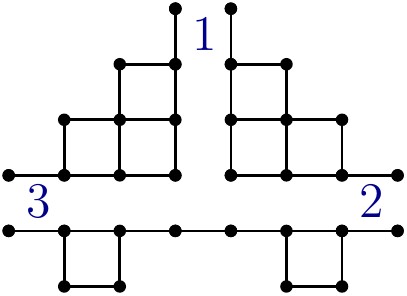}
		\caption{The clause polygon for square-smallness.}
		\label{fig:clause_small_old}
	\end{subfigure}\qquad
	\begin{subfigure}[t]{0.45\textwidth}
		\includegraphics[scale=0.9]{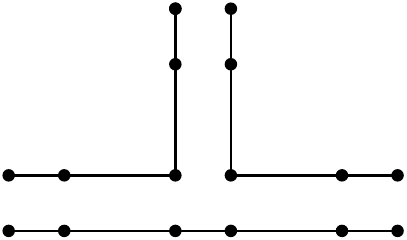}
		\caption{The clause polygon for disk-smallness.}
		\label{fig:clause_small_new}
	\end{subfigure}\medskip\\
	\begin{subfigure}[t]{0.45\textwidth}
		\includegraphics[scale=0.9]{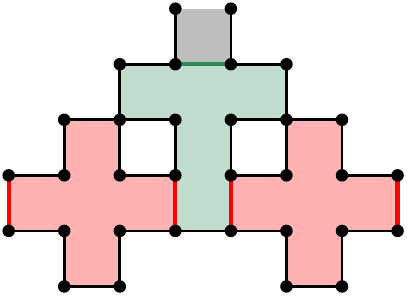}
		\caption{\texttt{True, False, False} for square-smallness.}
		\label{fig:clauseTFF_small_old}
	\end{subfigure}
    \qquad
	\begin{subfigure}[t]{0.45\textwidth}
		\includegraphics[scale=0.9]{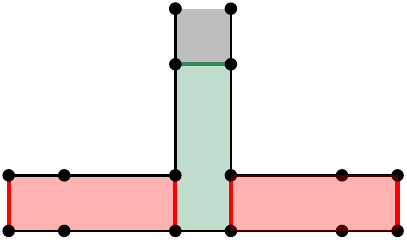}
		\caption{\texttt{True, False, False} for disk-smallness.}
		\label{fig:clauseTFF_small_new}
	\end{subfigure}
	\caption{Comparison of clause construction for square-smallness (left) and disk-smallness (right). Bottom: Partition of clause polygon for one truth assignment transmitted by wires.}
	\label{fig:constructclause_small}
\end{figure}

\end{document}